\documentclass[journal]{IEEEtran}

\usepackage{url}
\usepackage{amsmath,amsthm,amssymb,amsbsy,bbm}
\usepackage{mathdots}
\usepackage{paralist}
\usepackage{xcolor}
\usepackage{color}
\usepackage{graphicx}
\usepackage{algorithm,algpseudocode}
\usepackage{comment}

\newtheorem{thm}{Theorem}%[section] %(If you want theorem numbered
\newtheorem{cor}{Corollary}%[section]
\newtheorem{prop}{Proposition}%[section]
\newtheorem{defi}{Definition}%[section]

\newtheorem{lem}{Lemma}
\theoremstyle{remark}
\newtheorem{remark}{Remark}

\newcommand{\R}{\mathbb{R}}

\newcommand{\N}{\mathbb{N}}

\newcommand{\E}{\operatorname{E}}

\newcommand{\vct}[1]{\boldsymbol{#1}}
\newcommand{\mtx}[1]{\boldsymbol{#1}}
\newcommand{\T}{\mathrm{T}}

\newcommand{\trace}{\operatorname{trace}}

\newcommand{\rank}{\operatorname{rank}}

\DeclareMathOperator*{\minimize}{\text{minimize}}

\DeclareMathOperator*{\argmin}{\text{arg~min}}

\def \st {\operatorname*{subject\ to\ }}
\newcommand{\calA}{\mathcal{A}}

\newcommand{\calP}{\mathcal{P}}
\newcommand{\calC}{\mathcal{C}}
\newcommand{\calO}{\mathcal{O}}
\newcommand{\calE}{\mathcal{E}}

\newcommand{\vx}{\vct{x}}
\newcommand{\vy}{\vct{y}}

\newcommand{\bbmone}{\mathbbm{1}}

\newcommand{\mA}{\mtx{A}}
\newcommand{\mB}{\mtx{B}}
\newcommand{\mC}{\mtx{C}}
\newcommand{\mD}{\mtx{D}}

\newcommand{\mG}{\mtx{G}}
\newcommand{\mH}{\mtx{H}}

\newcommand{\mL}{\mtx{L}}

\newcommand{\mP}{\mtx{P}}
\newcommand{\mQ}{\mtx{Q}}
\newcommand{\mR}{\mtx{R}}

\newcommand{\mU}{\mtx{U}}
\newcommand{\mV}{\mtx{V}}
\newcommand{\mW}{\mtx{W}}
\newcommand{\mX}{\mtx{X}}
\newcommand{\mY}{\mtx{Y}}
\newcommand{\mZ}{\mtx{Z}}
\newcommand{\mDelta}{\mtx{\Delta}}

\newcommand{\mOmega}{\mtx{\Omega}}
\newcommand{\mPhi}{\mtx{\Phi}}
\newcommand{\mPsi}{\mtx{\Psi}}
\newcommand{\mSigma}{\mtx{\Sigma}}

\newcommand{\mId}{{\bf I}}

\newcommand{\mzero}{{\bf 0}}

\graphicspath{{./figs/}}

\newlength{\imgwidth}
\newboolean{twoColVersion}
\setboolean{twoColVersion}{false}
\newcommand{\twoCol}[2]{\ifthenelse{\boolean{twoColVersion}} {#1} {#2} }

\usepackage[english]{babel}

\title{Global Optimality in Low-rank Matrix Optimization}
\author{Zhihui Zhu, Qiuwei Li, Gongguo Tang, and Michael B. Wakin
\thanks{This work was supported by NSF grant CCF-1409261, NSF grant CCF-1464205, and NSF CAREER grant CCF-1149225.}
\thanks {Z. Zhu, Q. Li, G. Tang, and M. B. Wakin are with the Department of Electrical Engineering, Colorado School of Mines, Golden, CO 80401 USA. Email: \{zzhu, qiuli, gtang, mwakin\}@mines.edu.}
}
\begin{document}

\maketitle
\begin{abstract}
This paper considers the minimization of a general objective function $f(\mX)$ over the set of rectangular $n\times m$ matrices that have rank at most $r$. To reduce the computational burden, we factorize the variable $\mX$ into a product of two smaller matrices and optimize over these two matrices instead of $\mX$. Despite the resulting nonconvexity, recent studies in matrix completion and sensing have shown that the factored problem has no spurious local minima and obeys the so-called strict saddle property (the function has a directional negative curvature at all critical points but local minima). We analyze the global geometry for a general and yet well-conditioned objective function $f(\mX)$ whose restricted strong convexity and restricted strong smoothness constants are comparable. In particular, we show that the reformulated objective function has no spurious local minima and obeys the strict saddle property. These geometric properties imply that a number of iterative optimization algorithms (such as gradient descent) can provably solve the factored problem with global convergence.
\end{abstract}


\IEEEpeerreviewmaketitle
\section{Introduction}
Consider the minimization of a general objective function $f(\mX)$ over all low-rank $n\times m$ matrices:{
\begin{equation}\begin{split}
& \minimize_{\mX\in\R^{n\times m}} f(\mX)\\
& \st \rank(\mX)\leq r,
\label{eq:original problem}\end{split}\end{equation}
where the objective function $f:\R^{n\times m}\rightarrow \R$ is smooth.} Low-rank matrix optimizations of the form \eqref{eq:original problem} appear in a wide variety of applications, including quantum tomography \cite{aaronson2007learnability,flammia2012quantum}, collaborative filtering \cite{srebro2004maximum,decoste2006collaborative}, sensor localization \cite{biswas2004semidefinite}, low-rank matrix recovery from compressive measurements \cite{tang2011lower,recht2010guaranteed}, and matrix completion \cite{candes2009exact,liu2016low}. {Due to the rank constraint, however, low-rank matrix optimizations of the form~\eqref{eq:original problem} are highly nonconvex and computationally NP-hard in general~\cite{fazel2004rank} even if $f$ itself is convex.} In order to deal with the rank constraint and to find a low-rank solution, the nuclear norm is widely used in matrix inverse problems \cite{candes2011tight,recht2010guaranteed} arising in machine learning \cite{harchaoui2012large}, signal processing \cite{davenport2016overview},  and control \cite{mohan2010reweighted}. Although nuclear norm minimization enjoys strong statistical guarantees~\cite{candes2009exact}, its computational complexity is very high (as most algorithms require performing an expensive singular value decomposition (SVD) in each iteration), prohibiting it from scaling to practical problems.

To relieve the computational bottleneck {and provide an alternative way of dealing with the rank constraint,} recent studies propose to factorize the variable into the Burer-Monteiro type decomposition~\cite{burer2003nonlinear,burer2005local} with $\mX = \mU\mV^\T$, and optimize over the $n\times r$ and $m\times r$ matrices $\mU$ and $\mV$. With this parameterization of $\mX$, we can recast \eqref{eq:original problem} into the following program:
\begin{align}
\minimize_{\mU\in\R^{n\times r},\mV\in\R^{m\times r}} h(\mU,\mV):=f(\mU\mV^\T).
\label{eq:factored problem no regularizer}\end{align}
The bilinear nature of the parameterization renders the objective function of \eqref{eq:factored problem no regularizer} nonconvex even when $f(\mX)$ is a convex function. Hence, the objective function in~\eqref{eq:factored problem no regularizer} can  potentially have spurious local minima (i.e., local minimizers that are not global minimizers) or ``bad'' saddle points that prevent a number of iterative algorithms from converging to the global solution. By analyzing the landscape of nonconvex functions, several recent works have shown that  the factored objective function $h(\mU,\mV)$ in  certain matrix inverse problems has no spurious local minima~\cite{bhojanapalli2016lowrankrecoveryl,ge2016matrix,park2016non}.

We generalize this line of work by focusing on a general objective function $f(\mX)$ in the optimization \eqref{eq:original problem}, not necessarily a quadratic loss function coming from a matrix inverse problem. By focusing on a general objective function, we attempt to provide a unifying framework for low-rank matrix optimizations with the factorization approach. We provide a geometric analysis for the factored program~\eqref{eq:factored problem no regularizer} and show that, under certain conditions {on} $f(\mX)$, all critical points of the objective function $h(\mU,\mV)$ are well-behaved. Our characterization of the geometry of the objective function ensures that a number of iterative optimization algorithms converge to a global minimum.
%As also illustrated in Burer and Monterio's work \cite{burer2003nonlinear,burer2005local}, we remark that the factored problem \eqref{eq:factored problem no regularizer} serves as a way to solve the convex optimization \eqref{eq:original problem} globally, rather than a new modeling method.

\subsection{Summary of Results}
The purpose of this paper is to analyze the geometry of the factored problem $h(\mU,\mV)$ in~\eqref{eq:factored problem no regularizer}. In particular, we attempt to understand the behavior of all of the critical points of the objective function in the reformulated problem~\eqref{eq:factored problem no regularizer}.

Before presenting our main results, we lay out the necessary assumptions on the objective function $f(\mX)$. As is known, without any assumptions on the problem, even minimizing traditional quadratic objective functions is challenging. For this purpose, we focus on the model where
$f(\mX)$ is $(2r,4r)$-restricted strongly convex and smooth, i.e., for any $n\times m$ matrices $\mX, \mG$ with $\rank(\mX)\leq 2r$ and $\rank(\mG)\leq 4r$, the Hessian of $f(\mX)$ satisfies
% \begin{align}
%&\text{Restricted~Strong Conveixty}f(\mY) \geq f(\mX) + \left\langle \nabla(\mX), \mY - \mX \right\rangle - \frac{\alpha}{2}\left\|\mY - \mX\right\|_F^2\\
%&\text{Restricted~Strong Smoothness}f(\mY) \leq f(\mX) + \left\langle \nabla(\mX), \mY - \mX \right\rangle + \frac{\beta}{2}\left\|\mY - \mX\right\|_F^2
%\label{eq:RIP like}\end{align}
\begin{align}
\alpha\left\|\mG\right\|_F^2 \leq [\nabla^2 f(\mX)](\mG,\mG) \leq \beta \left\|\mG\right\|_F^2
\label{eq:RIP like}\end{align}
for some positive $\alpha$ and $\beta$.  A similar assumption is also utilized in~\cite[Conditions 5.3 and 5.4]{wang2016unified}. With this assumption on $f(\mX)$, we summarize our main results in the following informal theorem.

\begin{thm}{\em (informal)} Suppose the function $f(\mX)$ satisfies the $(2r,4r)$-restricted strong convexity and smoothness condition~\eqref{eq:RIP like} and has a critical point $\mX^\star\in\R^{n\times m}$ with $\rank(\mX^\star) = r^\star\leq r$. Then the factored objective function $h(\mU,\mV)$ (with an additional regularizer, see Theorem~\ref{thm:stricit saddle}) in \eqref{eq:factored problem no regularizer}  has no spurious local minima and obeys the strict saddle property (see Definition~\ref{def:strict saddle property} in Section~\ref{sec:Preliminaries}).
\label{thm:informal}\end{thm}

\begin{remark} {As guaranteed by Proposition~\ref{prop:RIP to unique} (in Section \ref{sec:main results}), the $(2r,4r)$-restricted strong convexity and smoothness property~\eqref{eq:RIP like} ensures that $\mX^\star$ is the unique global minimum of \eqref{eq:original problem}.}
Theorem~\ref{thm:informal} then implies that we can recover the rank-$r^\star$ global minimizer $\mX^\star$ of \eqref{eq:original problem} by many iterative algorithms (such as the trust region method~\cite{sun2016geometric} and stochastic gradient descent~\cite{ge2015escaping}) even from a random initialization. This is because 1) as guaranteed by Theorem~\ref{thm:global convergence}, the strict saddle property ensures local search algorithms converge to a local minimum, and 2) there are no spurious local minima.
\end{remark}
\begin{remark}
Since our main result only requires the $(2r,4r)$-restricted strong convexity and smoothness property~\eqref{eq:RIP like}, aside from low-rank matrix recovery~\cite{candes2011tight}, it can also be applied to many other low-rank matrix optimization problems~\cite{udell2014generalized} which do not necessarily involve quadratic loss functions. Typical examples include robust PCA~\cite{candes2011robust,bouwmans2016handbook}, 1-bit matrix completion~\cite{davenport20141,cai2013max} and Poisson principal component analysis (PCA)~\cite{salmon2014poisson}.
\end{remark}

\begin{remark}
{
Similar results on positive semi-definite (PSD) matrix optimization problems (but without the rank constraint) with generic objective functions were obtained in \cite{li2016}. We note that one cannot directly apply the results in \cite{li2016} to the optimization \eqref{eq:original problem} when the matrices under consideration are nonsymmetric or rectangular, even if we ignore the rank constraint. One could attempt to convert minimizing $f(\mX)$ over general $n\times m$ matrices into minimizing $q(\mZ)$ over the cone of PSD matrices of size $(m+n) \times (m+n)$,  where  $\mX$ and $\mX^\T$ form the upper right and lower left blocks of $\mZ$. The problem with this transformation, however, is that $q(\mZ)$ will no longer satisfy the same properties as $f(\mX)$, in particular the restricted strong convexity and smoothness condition~\eqref{eq:RIP like} which is a key assumption utilized in \cite{li2016}. For this reason, one cannot apply the results for the PSD optimization in \cite{li2016} directly to our problem. In terms of the proof techniques, although the generalization from the PSD case might not seem technically challenging at first sight, quite a few technical difficulties had to be overcome to develop the theory for the general case in this paper. In fact, the non-triviality of extending to the nonsymmetric case is also highlighted in \cite{tu2015low,park2016non}.
}
\end{remark}

\subsection{Related Works}
Compared with the original program~\eqref{eq:original problem},  the factored form~\eqref{eq:factored problem no regularizer} typically involves many fewer variables (or variables with much smaller size) and  can be efficiently solved by simple but powerful methods (such as gradient descent \cite{ge2015escaping,lee2016gradient}, the trust region method \cite{sun2015nonconvex}, and alternating methods \cite{jain2013low}) for large-scale settings, though it is nonconvex. In recent years, tremendous effort has been devoted to analyzing nonconvex optimizations by exploiting the geometry of the corresponding objective functions. These works can be separated into two types based on whether the geometry is analysed locally or globally. One type of work analyzes the behavior of the objective function in a small neighborhood containing the global optimum and {requires} a good initialization that is close enough to a global minimum. Problems such as phase retrieval~\cite{candes2015Wirtinger}, matrix sensing~\cite{tu2015low}, and semi-definite optimization~\cite{bhojanapalli2015dropping} have been studied.

Another type of work attempts to analyze the landscape of the objective function and show that it obeys the strict saddle property. If this particular property holds, then simple algorithms such as gradient descent and the trust region method are guaranteed to converge to a local minimum from a random initialization~\cite{lee2016gradient,ge2015escaping,sun2015complete} rather than requiring a good guess.  We approach low-rank matrix optimization with general objective functions~\eqref{eq:original problem} via a similar geometric characterization. Similar geometric results are known for a number of problems including complete dictionary learning~\cite{sun2015complete}, phase retrieval~\cite{sun2016geometric}, orthogonal tensor decomposition~\cite{ge2015escaping}, and matrix inverse problems~\cite{bhojanapalli2016lowrankrecoveryl,ge2016matrix,li2017geometry}. Empirical evidence also supports using the factorization approach for estimating a low-rank PSD matrix from a set of rank-one measurements corrupted by arbitrary outliers~\cite{li2017low} and for recovering a dynamically evolving low-rank matrix from incomplete observations~\cite{xu2016dynamic}.

Our work is most closely related to certain recent works in low-rank matrix optimization. Bhojanapalli et al.~\cite{bhojanapalli2016lowrankrecoveryl} showed that the low-rank, PSD matrix sensing problem has no spurious local minima and obeys the strict saddle property. Similar results were exploited for PSD matrix completion \cite{ge2016matrix}, PSD matrix factorization \cite{li2016symmetry} and low-rank, PSD matrix optimization problems with generic objective functions \cite{li2016}. Our work extends this line of analysis to general low-rank matrix (not necessary PSD or even square) optimization problems. Another closely related work considers the low-rank, non-square matrix sensing problem and matrix completion with the factorization approach~\cite{park2016non,zhu2017global,ge2017no}.  We note that our general objective function framework includes the low-rank matrix sensing problem as a special case (see Section~\ref{sec:Stylized Applications}). Furthermore, our result covers both over-parameterization where $r>r^\star$ and exact parameterization where $r=r^\star$. Wang et al.~\cite{wang2016unified} also considered the factored low-rank matrix minimization problem with a general objective function which satisfies the restricted strong convexity and smoothness condition. Their algorithms require good initializations for global convergence since they characterized only the local landscapes around the global optima. By categorizing the behavior of all the critical points, our work differs from \cite{wang2016unified} in that we instead characterize the global landscape of the factored objective function.

This paper continues in Section~\ref{sec:Preliminaries} with formal definitions for strict saddles and the strict saddle property. We present the main results and their implications in matrix sensing, weighted low-rank approximation, and 1-bit matrix completion in Section~\ref{sec:main results}. The proof of our main results is given in Section~\ref{sec:proof}. We conclude the paper in Section~\ref{sec:conclusion}.

\section{Preliminaries}\label{sec:Preliminaries}
\subsection{Notation}
To begin, we first briefly introduce some notation used throughout the paper. The symbols $\mId$ and $\mzero$ respectively represent the identity matrix and zero matrix with appropriate sizes. The set of $r\times r$ orthonormal matrices is denoted by $\calO_r:=\{\mR\in\R^{r\times r}:\mR^\T\mR = \mId\}$. If a function $h(\mU,\mV)$ has two arguments, $\mU\in\R^{n\times r}$ and $\mV\in\R^{m\times r}$, we occasionally use the notation $h(\mW)$ when we put these two arguments into a new one as $\mW=\begin{bmatrix}\mU \\ \mV \end{bmatrix}$.  For a scalar function $f(\mZ)$ with a matrix variable $\mZ\in\R^{n\times m}$, its gradient is an $n\times m$ matrix whose $(i,j)$-th entry is $[\nabla f(\mZ)]_{ij} = \frac{\partial f(\mZ)}{\partial Z_{ij}}$ for all $i\in [n], j\in [m] $. Here $[n] = \{1,2,\ldots,n\}$ for any $n\in \N$ and $Z_{ij}$ is the $(i,j)$-th entry of the matrix $\mZ$.
The Hessian of $f(\mZ)$ can be viewed as an $nm\times nm$ matrix $[\nabla^2 f(\mZ)]_{ij} = \frac{\partial^2 f(\mZ)}{\partial z_{i}\partial z_j}$ for all $i,j\in [nm]$, where $z_i$ is the $i$-th entry of the vectorization of $\mZ$. An alternative way to represent the Hessian is by a bilinear form defined via $[\nabla^2f(\mZ)](\mA,\mB) = \sum_{i,j,k,l}\frac{\partial^2 f(\mZ)}{\partial Z_{ij}\partial Z_{kl}} A_{ij} B_{kl}$ for any $\mA,\mB\in\R^{n\times m}$. The bilinear form for the Hessian is widely utilized through the paper.

\subsection{Strict Saddle Property}
Suppose $h:\R^{n}\rightarrow \R$ is a twice continuously differentiable objective function. We begin with the notion of strict saddles and the strict saddle property.
\begin{defi}[Critical points]
We say $\vx$ a critical point if the gradient at $\vx$ vanishes, i.e., $\nabla h(\vx) = \mzero$.
\end{defi}

\begin{defi}[Strict saddles]
A critical point $\vx$ is a strict saddle if the Hessian matrix evaluated at this point has a strictly negative eigenvalue, i.e., $\lambda_{\min}(\nabla^2 h(\vx))<0$.
\end{defi}

\begin{defi}[Strict saddle property~\cite{ge2015escaping}]\label{def:strict saddle property}
A twice differentiable function satisfies the strict saddle property if each critical point either corresponds to a local minimum or is a strict saddle.
\end{defi}
Intuitively, the strict saddle property requires a function to have a directional negative curvature at all critical points but local minima. This property allows a number of iterative algorithms such as noisy gradient descent~\cite{ge2015escaping} and the trust region method~\cite{conn2000trust} to further decrease the function value at all the strict saddles and thus converge to a local minimum.

\begin{thm}\label{thm:global convergence}
\cite{sun2015nonconvex,ge2015escaping,lee2016gradient} (informal)
For a twice continuously differentiable objective function satisfying the strict saddle property, a number of iterative optimization algorithms (such as gradient descent and the the trust region method) can find a local minimum.
\end{thm}

\section{Problem Formulation and Main Results} \label{sec:main results}
\subsection{Problem Formulation}
{This paper considers the problem \eqref{eq:original problem} of minimizing a general function $f(\mX)$ (over the set of low-rank matrices) which is assumed to have a low-rank critical point $\mX^\star$ with $\rank(\mX^\star) = r^\star \leq r$ such that  $\nabla f(\mX^\star) = \mzero$. Because of the restricted strong convexity and smoothness condition \eqref{eq:RIP like}, the following result establishes that if $f(\mX)$  has a critical point $\mX^\star$ with $\rank(\mX^\star)\leq r$, then it is the unique global minimum of \eqref{eq:original problem}.
\begin{prop}\label{prop:RIP to unique} Suppose $f(\mX)$ satisfies the $(2r,4r)$-restricted  strong convexity and smoothness condition \eqref{eq:RIP like} with positive  $\alpha$ and $\beta$. {Assume $\mX^\star$ is a critical point of $f(\mX)$ with $\rank(\mX^\star) = r^\star \leq r$. Then $\mX^\star$ is the global minimum of \eqref{eq:original problem}, i.e.,
\[
f(\mX^\star)\leq f(\mX), \ \forall \ \mX\in\R^{n\times m}, \rank(\mX)\leq r
\]
and the equality holds only at $\mX = \mX^\star$.}
\end{prop}
\begin{proof}[Proof of Proposition~\ref{prop:RIP to unique}]
First note that if $\mX^\star$ is a critical point of $f(\mX)$, then
\[
\nabla f(\mX^\star) = \mzero.
\]
Now for any $\mX\in\R^{n\times m}$ with $\rank(\mX)\leq r$, the second order Taylor expansion gives
\begin{align*}
f(\mX) =& f(\mX^\star) + \left\langle \nabla f(\mX^\star), \mX-\mX^\star \right\rangle \\
&+ \frac{1}{2}[\nabla^2 f(\widetilde \mX)](\mX-\mX^\star, \mX-\mX^\star),
\end{align*}
where $\widetilde \mX = t\mX^\star + (1-t)\mX$ for some $t\in[0,1]$. This Taylor expansion together with  $\nabla f(\mX^\star) = \mzero$ and \eqref{eq:RIP like} (both $\widetilde\mX$ and  $\mX'-\mX^\star$ have rank at most $2r$) gives
\begin{align*}
f(\mX) - f(\mX^\star) & = \frac{1}{2}[\nabla^2 f(\widetilde \mX)](\mX-\mX^\star, \mX-\mX^\star)\\ & \geq \frac{\alpha}{2}\|\mX - \mX^\star\|_F^2.
\end{align*}
\end{proof}
With this, in the sequel, we use $\mX^\star$ to denote the global minimum of \eqref{eq:original problem} (i.e., the low-rank critical point of $f(\mX)$), unless stated otherwise. We note that the assumption of the existence of a low-rank critical point $\mX^\star$ is very mild and holds in many matrix inverse problems \cite{recht2010guaranteed,candes2009exact}, where the unknown matrix to be recovered is a critical point of $f$.
}
We factorize the variable $\mX = \mU\mV^\T$ with $\mU\in\R^{n\times r}, \mV\in\R^{m\times r}$ and transform~\eqref{eq:original problem} into its factored counterpart \eqref{eq:factored problem no regularizer}.  Throughout the paper, $\mX$, $\mW$ and $\widehat\mW$ are matrices depending on $\mU$ and $\mV$:
\[
\mW = \begin{bmatrix} \mU \\ \mV \end{bmatrix}, ~ \widehat\mW = \begin{bmatrix} \mU \\ -\mV \end{bmatrix},~ \mX = \mU\mV^\T.
\]
Although the new variable  $\mW$ has much smaller size than $\mX$ when $r\ll \min\{n,m\}$, the objective function in the factored problem \eqref{eq:factored problem no regularizer} may have a much more complicated landscape due to the bilinear form about $\mU$ and $\mV$.  The reformulated objective function $h(\mU,\mV)$ could introduce spurious local minima or degenerate saddle points even when $f(\mX)$ is convex. Our goal is to guarantee that this does not happen.

Let $\mX^\star = \mQ_{\mU^\star}\mSigma^\star\mQ_{\mV^\star}^\T$ denote an SVD of $\mX^\star$, where $\mQ_{\mU^\star}\in\R^{n\times r}$ and $\mQ_{\mV^\star}\in\R^{m\times r}$ are orthonormal matrices of appropriate sizes, and $\mSigma^\star\in\R^{r\times r}$ is a diagonal matrix with non-negative diagonal (but with some zeros on the diagonal if $r>r^\star = \rank(\mX^\star)$). We denote
\[
\mU^\star = \mQ_{\mU^\star}{\mSigma^\star}^{1/2}, \quad \mV^\star = \mQ_{\mV^\star}{\mSigma^\star}^{1/2},
\]
where $\mX^\star = \mU^\star\mV^{\star\T}$ forms a balanced factorization of $\mX^\star$ since $\mU^\star$ and $\mV^\star$ have the same singular values. Throughout the paper, we utilize the following two ways to stack $\mU^\star$ and $\mV^\star$ together:
\[
\mW^\star = \begin{bmatrix} \mU^\star \\ \mV^\star \end{bmatrix}, \quad \widehat\mW^\star = \begin{bmatrix} \mU^\star \\ -\mV^\star \end{bmatrix}.
\]
%One useful consequence of this notation is as follows
%\begin{align}
%\widehat \mW^\star\widehat\mW^{\star\T} \mW^\star{\mW^\star}^\T = \mzero.
%\label{eq:eaual foot for Wstar}\end{align}
Before moving on, we note that for any solution $(\mU,\mV)$ to~\eqref{eq:factored problem no regularizer}, $(\mU\mPsi,\mV\mPhi)$ is also a solution to \eqref{eq:factored problem no regularizer} for any $\mPsi,\mPhi\in\R^{r\times r}$ such that $\mU\mPsi\mPhi^\T\mV^\T = \mU\mV^\T$. In order to address this ambiguity (i.e., to reduce the search space of $\mW$ for~\eqref{eq:factored problem no regularizer}), we utilize the trick in \cite{tu2015low,park2016non,wang2016unified} by introducing a regularizer
\begin{align}
g(\mU,\mV) = \frac{\mu}{4}\left\|\mU^\T\mU - \mV^\T \mV\right\|_F^2
\label{eq:define g}\end{align}
 and  solving the following problem
\begin{align}
\minimize_{\mU\in\R^{n\times r},\mV\in\R^{m\times r}} \rho(\mU,\mV):= f(\mU\mV^\T) + g(\mU,\mV),
\label{eq:factored problem}\end{align}
where $\mu>0$ controls the weight for the term $\left\|\mU^\T\mU - \mV^\T \mV\right\|_F^2$, which will be discussed soon.

We remark that $\mW^\star$ is still a global minimizer {of} the factored problem \eqref{eq:factored problem} since $f(\mX)$ achieves its  global minimum over the low-rank set of matrices at $\mX^\star$ and $g(\mW)$ also achieves its  global minimum at $\mW^\star$. The regularizer $g(\mW)$ is applied to force the difference between the two Gram matrices of $\mU$ and $\mV$ to be as small as possible. The global minimum of $g(\mW)$ is $0$, which is achieved when $\mU$ and $\mV$ have the same Gram matrices, i.e., when $\mW$ belongs to
\begin{align}\label{eq:set of balanced factors}
\calE: = \left\{\mW = \begin{bmatrix} \mU \\ \mV \end{bmatrix}: \mU^\T\mU - \mV^\T \mV =  \mzero\right\}.
\end{align}
Informally, we can view \eqref{eq:factored problem} as finding a point from $\calE$ that also minimizes $f(\mU\mV^\T)$. This is formally established in Theorem~\ref{thm:stricit saddle}.

\subsection{Main Results}
Our main argument is that, under certain conditions on $f(\mX)$, the objective function $\rho(\mW)$ has no spurious local minima and satisfies the strict saddle property.  This is equivalent to categorizing all the critical points into two types: 1) the global minima which correspond to the global solution of the original convex problem \eqref{eq:original problem} and 2) strict saddles such that the Hessian matrix $\nabla^2\rho(\mW)$ evaluated at these points has a strictly negative eigenvalue. We formally establish this in the following theorem, whose proof is given in the next section.

\begin{thm}   For any $\mu>0$, each critical point $\mW= \begin{bmatrix} \mU \\ \mV \end{bmatrix}$ of $\rho(\mW)$ defined in~\eqref{eq:factored problem} satisfies
\begin{align}\label{eq:thm eq 1}
\mU^\T\mU - \mV^\T\mV = \mzero.
\end{align}
Furthermore, suppose that the function $f(\mX)$ satisfies the $(2r,4r)$-restricted strong convexity and smoothness condition \eqref{eq:RIP like} with positive constants $\alpha$ and $\beta$ satisfying $\frac{\beta}{\alpha}\leq 1.5$ and that the function $f(\mX)$ has a critical point $\mX^\star\in\R^{n\times m}$ with $\rank(\mX^\star) = r^\star\leq r$. Set $\mu \leq \frac{\alpha}{16}$ for the factored problem \eqref{eq:factored problem}. Then $\rho(\mW)$ has no spurious local minima, i.e., any local minimum of $\rho(\mW)$ is a global minimum corresponding to the global solution of the original problem \eqref{eq:original problem}:
$
\mU\mV^\T = \mX^\star.
$
In addition, $\rho(\mW)$ obeys the strict saddle property that any critical point not being a local minimum is a strict saddle with
\begin{equation}\begin{split}
&\lambda_{\min}\left(\nabla^2\left(\rho(\mW)\right) \right)\leq \\
&\left\{\begin{matrix} -0.08\alpha\sigma_{r}(\mX^\star), & r = r^\star \\
-0.05\alpha\cdot \min\left\{\sigma_{r^c}^2(\mW),2\sigma_{r^\star}(\mX^\star)\right\}, & r>r^\star\\
 -0.1\alpha \sigma_{r^\star}(\mX^\star), & r_c = 0,
\end{matrix}\right.
\label{eq:strict saddle}\end{split}\end{equation}
where $r^c\leq r$ is the rank of $\mW$, $\lambda_{\min}(\cdot)$ represents the smallest eigenvalue, and $\sigma_\ell(\cdot)$ denotes the $\ell$-th largest singular value.
\label{thm:stricit saddle}\end{thm}
\begin{remark}
Equation~\eqref{eq:thm eq 1} shows that any critical point $\mW$ belongs to $\calE$ for the objective function in the factored problem~\eqref{eq:factored problem} with any positive $\mu$. This
demonstrates the reason for adding the regularizer $g(\mU,\mV)$. Thus, any iterative optimization algorithm converging to some critical point of $\rho(\mW)$ results in a solution within $\calE$. Furthermore, the strict saddle property along with the lack of spurious local minima ensures that a number of iterative optimization algorithms find the global minimum.
\end{remark}
\begin{remark}
For any critical point $\mW\in\R^{(n+m)\times r}$ that is not a local minimum, the right hand side of \eqref{eq:strict saddle} is strictly negative, implying {that} $\mW$ is a strict saddle. We also note that Theorem~\ref{thm:stricit saddle} not only covers exact parameterization where $r=r^\star$, but also includes the over-parameterization {case} where $r>r^\star$.
\end{remark}
\begin{remark}
The constants appearing in Theorem~\ref{thm:stricit saddle} are not optimized. We use $\mu\leq \frac{1}{16} \alpha$ simply to include $\mu = \frac{1}{16}$ which is utilized for the matrix sensing problem in~\cite{tu2015low}. If the ratio between the restricted strong convexity and smoothness constants $\frac{\beta}{\alpha}\leq 1.4$, then we can show that $\rho(\mW)$ has no spurious local minima and obeys the strict saddle property for any $\mu\leq \frac{1}{4}\alpha$ (where $\mu = \frac{1}{4}$ is utilized for the matrix sensing problem in~\cite{park2016non}). In all cases, a smaller $\mu$ yields a more negative constant in~\eqref{eq:strict saddle}; see Section~\ref{sec:proof} for more discussion on this. This implies that when the restricted strong convexity constant $\alpha$ is not provided a priori, one can always choose a small $\mu$ to ensure the strict saddle property holds, and hence guarantee the global convergence of many iterative optimization algorithms.

{The constant $1.5$ for the dynamic range $\frac{\beta}{\alpha}$ in Theorem~\ref{thm:stricit saddle} is also not optimized and it is possible to slightly relax this constraint with more sophisticated analysis. However, the following example involving weighted symmetric matrix factorization implies that the room for improving this constant is rather limited. Let
\[
\mOmega=\begin{bmatrix}\sqrt{1+a}&1\\1&\sqrt{1+a}\end{bmatrix}
\]
for some $a\geq0$,
\[
\mX^\star=\begin{bmatrix}1&1\\1&1\end{bmatrix}, ~\text{and}~ \mU=\begin{bmatrix}x\\y\end{bmatrix}.
\]
Now consider the following weighted low-rank matrix factorization:
	\begin{equation}\begin{split}
	& h(\mU)= \frac{1}{2}\|\mOmega\odot(\mU\mU^\T-\mX^\star)\|_F^2 \\ &=\frac{1+a}{2}\left(x^2-1\right)^2+\frac{1+a}{2}\left(y^2-1\right)^2+  (x y-1)^2,
\end{split}\label{eq:symmetric wpca}	\end{equation}
whose gradient $\nabla h(\mU)$ and Hessian $\nabla^2 h(\mU)$ are given by:
\[
		\nabla h(\mU)
		=2\begin{bmatrix}
		(a+1) \left(x^2-1\right) x+  y (x y-1)\\
		(a+1) \left(y^2-1\right) y+  x (x y-1)
		\end{bmatrix},
\]
and
\begin{align*}
&		\nabla^2h(\mU)= \\
&\quad		2\begin{bmatrix}
		y^2+\left(3 x^2-1\right) (a+1) & 2 x y-1 \\
		2 x y-1 &   x^2+\left(3 y^2-1\right) (a+1)\\
		\end{bmatrix}.
\end{align*}
Then,
\[
\mU = \begin{bmatrix}\sqrt{\frac{a}{a+2}}\\  -\sqrt{\frac{a}{a+2}}\end{bmatrix}
\]
is a critical point with
\[
\nabla^2 h(\mU)=
			\begin{bmatrix}
			4 a+\frac{8}{a+2}-6 & \frac{8}{a+2}-6 \\
			\frac{8}{a+2}-6 & 4 a+\frac{8}{a+2}-6 \\
			\end{bmatrix},
\]
which has eigenvalues
\[
\lambda_1=\frac{4 (a-2) (a+1)}{a+2}\begin{cases}<0, & a\in[0,2),\\ >0, &a>2, \end{cases}
\]
and $\lambda_2=4 a>0$. We conclude that this $\mU$ is a strict saddle point when $a<2$ and a spurious local minimum when $a>2$. This weighted symmetric matrix factorization problem \eqref{eq:symmetric wpca} satisfies the restricted strong convexity and smoothness condition~\eqref{eq:RIP like} with constants $\alpha = \|\mOmega\|_{\min}^2 = 1$ and $\beta = \|\mOmega\|_{\max}^2 =1+a$ (where $\|\mOmega\|_{\min}$ and $\|\mOmega\|_{\max}$ represent the smallest and largest entries in $\mOmega$; see Section \ref{sec:Stylized Applications}). Thus, we have a counter example which demonstrates the existence of spurious local minima when $\frac{\beta}{\alpha}>3$.}
\end{remark}

\begin{remark}
We finally remark that although Theorem~\ref{thm:stricit saddle} requires the additional regularizer \eqref{eq:define g}, empirical evidence (see experiments in Section~\ref{sec:experiments}) shows we can get rid of this regularizer for many iterative algorithms with random initialization.
\end{remark}

We prove Theorem~\ref{thm:stricit saddle} in Section~\ref{sec:proof}. Before proceeding, we present two stylized applications of Theorem~\ref{thm:stricit saddle} in matrix sensing and weighted low-rank approximation.

\subsection{Stylized Applications}\label{sec:Stylized Applications}
\subsubsection{Matrix Sensing}
We first consider the implication of Theorem~\ref{thm:stricit saddle} in the matrix sensing problem where
\begin{align*}
 f(\mX) = \frac{1}{2}\left\|\calA\left(\mX - \mX^\star\right)\right\|_2^2.
\end{align*}
Here $\calA:\R^{n\times m}\rightarrow \R^p$ is a known measurement operator satisfying the following restricted isometry property.

\begin{defi}(Restricted Isometry Property (RIP)~\cite{recht2010guaranteed}) The map $\calA:\R^{n\times m}\rightarrow \R^p$ satisfies the $r$-RIP with constant $\delta_r$ if
\begin{align}\label{eq:RIP}
\left(1 - \delta_r\right) \left\|\mX\right\|_F^2 \leq \left\|\calA(\mX)\right\|^2\leq \left(1+\delta_r\right) \left\|\mX\right\|_F^2
\end{align}
holds for any $n\times m$ matrix $\mX$ with $\rank(\mX)\leq r$.
\end{defi}
{Note that, in this case, the gradient of $f(\mX)$ at $\mX^\star$ is
\[
\nabla f(\mX^\star) = \calA^*\calA(\mX^\star - \mX^\star) = \mzero,
\]
which implies that $\mX^\star$ is a critical point of $f(\mX)$.}
The Hessian quadrature form $\nabla^2 f(\mX)[\mY,\mY]$ for any $n\times m$ matrices $\mX$ and $\mY$ is given by
\[
 \nabla^2 f(\mX)[\mY,\mY] = \left\|\calA(\mY)\right\|^2.
\]
If $\calA$ satisfies the $4r$-restricted isometry property with constant $\delta_{4r}$, then $f(\mX)$ satisfies the $(2r,4r)$-restricted strong convexity and smoothness condition~\eqref{eq:RIP like} with constants $\alpha = 1-\delta_{4r}$ and $\beta = 1-\delta_{4r}$ since
\begin{align*}
(1-\delta_{4r})\left\|\mY\right\|_F^2 \leq \left\|\calA(\mY)\right\|^2\leq (1+\delta_{4r})\left\|\mY\right\|_F^2
\end{align*}
for any rank-$4r$ matrix $\mY$. Now, applying Theorem~\ref{thm:stricit saddle}, we can characterize the geometry for the following matrix sensing problem with the factorization approach:
\begin{align}
\minimize_{\mU\in\R^{n\times r},\mV\in \R^{n\times r}}  \frac{1}{2}\left\|\calA(\mU\mV^\T - \mX^\star)\right\|_2^2 + g(\mU,\mV),
\label{eq:matrix sensing}\end{align}
where $g(\mU,\mV)$ is the added regularizer defined in~\eqref{eq:define g}.
\begin{cor}
Suppose $\calA$ satisfies the $4r$-RIP with constant $\delta_{4r}\leq \frac{1}{5}$, and set $\mu \leq \frac{1-\delta_{4r}}{16}$. Then the objective function in~\eqref{eq:matrix sensing} has no spurious local minima and satisfies the strict saddle property.
\label{cor:matrix sensing}\end{cor}
This result follows directly from Theorem~\ref{thm:stricit saddle} by noting that
$\frac{\beta}{\alpha}= \frac{1+\delta_{4r}}{1-\delta_{4r}}\leq 1.5$ if $\delta_{4r}\leq \frac{1}{5}$. We remark that Park et al.~\cite[Theorem 4.3]{park2016non} provided a similar geometric result for~\eqref{eq:matrix sensing}. Compared to their result which requires $\delta_{4r}\leq \frac{1}{100}$, our result has a much weaker requirement on the RIP of the measurement operator.

\subsubsection{Weighted Low-Rank Matrix Factorization}
We now consider the implication of Theorem~\ref{thm:stricit saddle} in the weighted matrix factorization problem~\cite{srebro2003weighted}, where
\begin{align*}
 f(\mX) &:= \frac{1}{2}\left\|\mOmega\circ\left(\mX - \mX^\star\right)\right\|_F^2.
\end{align*}
Here $\mOmega$ is an $n\times m$ weight matrix consisting of positive elements and $\circ$ denotes the point-wise product between two matrices. { In this case, the gradient of $f(\mX)$ at $\mX^\star$ is
\[
\nabla f(\mX^\star) = \mOmega\circ\mOmega \circ(\mX^\star - \mX^\star) = \mzero,
\]
which implies that $\mX^\star$ is a critical point of $f(\mX)$.} The Hessian quadrature form $\nabla^2 f(\mX)[\mY,\mY]$ for any $n\times m$ matrices $\mX$ and $\mY$ is given by
\[
 \nabla^2 f(\mX)[\mY,\mY] = \left\|\mOmega \circ \mY\right\|_F^2.
\]
Thus $f(\mX)$ satisfies the $(2r,4r)$-restricted strong convexity and smoothness condition~\eqref{eq:RIP like} with constants $\alpha = \|\mOmega\|_{\min}^2$ and $\beta = \|\mOmega\|_{\max}^2$ since
\[
\left\|\mOmega\right\|_{\min}^2 \left\|\mY\right\|_F^2 \leq \left\|\mOmega \circ \mY\right\|_F^2  \leq \left\|\mOmega\right\|_{\max}^2 \left\|\mY\right\|_F^2,
\]
where $\|\mOmega\|_{\min}$ and $\|\mOmega\|_{\max}$ represent the smallest and largest entries in $\mOmega$, respectively.
Now we consider the following weighted matrix factorization problem:
\begin{align}
\minimize_{\mU\in\R^{n\times r},\mV \in\R^{n\times r}} \frac{1}{2} \left\|\mOmega \circ(\mU\mV^\T - \mX^\star)\right\|_F^2 + g(\mU,\mV),
\label{eq:weighted matrix fact}\end{align}
where $g(\mU,\mV)$ is the added regularizer defined in~\eqref{eq:define g}.
For an arbitrary weight matrix $\mOmega$, it is proven that the weighted low-rank factorization can be NP-hard~\cite{gillis2011low} and has spurious local minima.
%It is still an open problem  to efficiently obtain an optimal solution to
%the above problem  for any weight matrices $\mOmega$.
When the elements in the weight matrix $\mOmega$ are concentrated, it is expected that~\eqref{eq:weighted matrix fact} can be efficiently solved by a number of iterative optimization algorithms as it is close to an (unweighted) matrix factorization problem (where $\mOmega$ is a matrix of ones) which obeys the strict saddle property~\cite{li2016symmetry}. The following result characterizes the geometric structure in the objection function of~\eqref{eq:weighted matrix fact} by directly applying Theorem~\ref{thm:stricit saddle}.
\begin{cor}
Suppose $\mOmega$ satisfies $\frac{\|\mOmega\|_{\max}^2}{\|\mOmega\|_{\min}^2}\leq 1.5$. Set $\mu \leq \frac{\|\mOmega\|_{\min}^2}{16}$. Then the objective function in~\eqref{eq:weighted matrix fact} has no spurious local minima and satisfies the strict saddle property.
\end{cor}

\subsubsection{1-bit Matrix Completion}
Finally, we consider the problem of completing a low-rank matrix from a subset of 1-bit measurements~\cite{davenport20141}. Given $\mX^\diamond\in\R^{n\times m}$, a subset of indices $\Omega\subset[m]\times [n]$, and a differentiable function $q:\R \rightarrow [0,1]$, we observe
\begin{align}
Y_{i,j} = \left\{\begin{array}{ll} +1 & \text{with~probability} \ q(X^\diamond_{i,j}), \\ -1 &  \text{with~probability} \ 1- q(X^\diamond_{i,j}), \end{array}\right.
\label{eq:1bit observ}\end{align}
for all $(i,j)\in \Omega$. Typical choices for $q$ include the logistic regression model where $q(x) = \frac{e^{x}}{ 1+ e^x}$ and the probit regression model where $q(x) = 1 - \Phi(-x/\sigma) = \Phi(x/\sigma)$. Here $\Phi$ is the cumulative distribution function (CDF) of a mean-zero Gaussian distribution with variance $\sigma^2$.  In \cite{davenport20141}, the authors attempt to recover $\mX^\diamond$ from the incomplete nonlinear measurements $\{Y_{ij}\}_{(i,j)\in\Omega}$ by minimizing the negative log-likelihood function
\begin{align*}
F_{\Omega,\mY}(\mX): = -\sum_{(i,j)\in\Omega}\big(&\bbmone_{(Y_{i,j}=1)}\log(q(X_{i,j})) \\ & + \bbmone_{(Y_{i,j}=-1)}\log(1 - q(X_{i,j}))\big)
\end{align*}
which results in a maximum likelihood (ML) estimate.

We note that $F_{\Omega,\mY}$ is a convex function for both the logistic model and the probit model. The following result also establishes that $F_{\Omega,\mY}$ satisfies the restricted strong convexity and smoothness condition if we observe full 1-bit measurements, i.e., $\Omega = [n]\times [m]$.
\begin{lem} %Suppose $q(x)$ satisfies $q(x)+ q(1-x) = 1$.
Suppose $\Omega = [n]\times [m]$.
Let \begin{align*}
\alpha_{q,\gamma} = \min_{|x|\leq \gamma} \min\bigg(&\frac{(q'(x))^2 - q(x)q''(x)}{q^2(x)},\\& \frac{(q'(x))^2 +(1- q(x))q''(x)}{(1-q(x))^2} \bigg)
\end{align*}
 and
\begin{align*}
\beta_{q,\gamma} = \max_{|x|\leq \gamma} \max\bigg(&\frac{(q'(x))^2 - q(x)q''(x)}{q^2(x)},\\& \frac{(q'(x))^2 +(1- q(x))q''(x)}{(1-q(x))^2} \bigg).
\end{align*}
Then $F_{\Omega,\mY}$ satisfies the restricted strong convexity and smoothness condition:
\[
\alpha_{q,\gamma}\|\mG\|_F^2 \leq [\nabla^2 F_{\Omega,\mY}(\mX)](\mG,\mG) \leq \beta_{q,\gamma} \|\mG\|_F^2
\]
for any $\mG\in \R^{n\times m}$ and $\|\mX\|_\infty\leq \gamma$.
\label{lem:1bit}\end{lem}
The proof of Lemma~\ref{lem:1bit} is given in Appendix~\ref{sec:prf 1bit}. Now we {consider} the logistic regression model where $q(x) = \frac{e^x}{1+e^x}$.
\begin{cor} Suppose $\Omega = [n]\times [m]$ and $\gamma \leq 1.3$. Consider the logistic regression model where $q(x) = \frac{e^x}{1+e^x}$. Then $F_{\Omega,\mY}$ satisfies the restricted strong convexity and smoothness condition with
\[
\frac{\beta_{q,\gamma}}{\alpha_{q,\gamma}} \leq 1.5.
\]
\label{cor:1bit logistic}\end{cor}
\begin{proof}[Proof of Corollary~\ref{cor:1bit logistic}] Applying Lemma~\ref{lem:1bit} with direct calculation gives
\begin{align*}
&\alpha_{q,\gamma} = q'(\gamma) =  \frac{e^\gamma}{(1+e^\gamma)^2},\\
&\beta_{q,\gamma} = q'(0) =  \frac{e^0}{(1+e^0)^2} = \frac{1}{4},
\end{align*}
where $q'(x) = \frac{e^x}{(1 + e^x)^2}$. Now if we restrict $\|\mX\|_\infty \leq 1.3$, we have
\[
\frac{\beta_{q,\gamma}}{\alpha_{q,\gamma}} = 4 \frac{e^{1.3}}{(1+e^{1.3})^2} \leq 1.5.
\]
\end{proof}
Under the assumption that $\mX^\diamond$ is low-rank, a nuclear norm constraint is utilized in \cite{davenport20141} to force a low-rank solution. Corollary \ref{cor:1bit logistic} implies that we can apply matrix factorization for 1-bit matrix recovery given that the elements of $\mX$ are bounded. { For the setting where $\Omega$ is only a subset of $[n]\times [m]$, \cite{bhaskar20151} considered the 1-bit matrix {\em completion} problem with the rank constraint and established a stronger statistical recovery guarantee than that in \cite{davenport20141}. Empirical evidence (see \cite{bhaskar20151} and Section~\ref{sec:exp 1bit}) supports that matrix factorization also works for 1-bit matrix completion.
}

\section{Proof of Theorem \ref{thm:stricit saddle}}\label{sec:proof}
In this section, we provide a formal proof of Theorem~\ref{thm:stricit saddle}. The main argument involves showing that each critical point of $\rho(\mW)$ either corresponds to the global solution of~\eqref{eq:original problem} or is a strict saddle whose Hessian $\nabla^2 \rho (\mW)$ has a strictly negative eigenvalue. %Following the same proof techniques of~\cite{bhojanapalli2016lowrankrecoveryl,park2016non,li2016},
Specifically, we show that $\mW$ is a strict saddle by arguing that the Hessian $\nabla^2 \rho (\mW)$ has a strictly negative curvature along $\mDelta:= \mW -\mW^\star\mR$, i.e., $[\nabla^2 \rho (\mW)](\mDelta,\mDelta)\leq -\tau \|\mDelta\|_F^2$ for some $\tau>0$. Here $\mR$ is an $r\times r$ orthonormal matrix such that the distance between $\mW$ and $\mW^\star$ rotated through $\mR$ is as small as possible. %We note that for any $\mU\in\R^{n\times r}$ and $\mV\in\R^{m\times r}$, $\mW$, $\widehat\mW$ and $\mX$ are defined as follows
%\[
%\mW = \begin{bmatrix} \mU \\ \mV \end{bmatrix}, \quad \widehat\mW = \begin{bmatrix} \mU \\ -\mV \end{bmatrix},  \quad \mX = \mU\mV^\T.
%\]

\subsection{Supporting Results}
We first present some useful results. The $(2r,4r)$-restricted strong convexity and smoothness assumption \eqref{eq:RIP like} implies the following isometry property, whose proof is given in Appendix~\ref{sec:prf RIP like}.
\begin{prop}
Suppose the function $f(\mX)$ satisfies the $(2r,4r)$-restricted  strong convexity and smoothness condition \eqref{eq:RIP like} with positive $\alpha$ and $\beta$. Then for any $n\times m$ matrices $\mZ,\mG,\mH$ of rank at most $2r$, we have
\begin{align*}
&\left|\frac{2}{\alpha+\beta}[\nabla^2f(\mZ)](\mG,\mH) - \langle \mG,\mH \rangle\right|\leq \frac{\beta - \alpha}{\beta + \alpha}\left\|\mG\right\|_F \left\|\mH\right\|_F.
\end{align*}
\label{prop:RIP like}\end{prop}
%\begin{proof}[Proof of Proposition~\ref{prop:RIP like}]
%First note that the $(2r,4r)$-restricted  strong convexity and smoothness condition \eqref{eq:RIP like} implies
%\[
%\left|\left[\frac{2}{\alpha+\beta}\nabla^2f(\mZ) - \mId\right](\mG,\mG)\right| \leq \frac{\beta - \alpha}{\beta + \alpha}\left\|\mG\right\|_F^2.
%\]
%Thus we have
%\begin{align*}
%\left|\left[\frac{2}{\alpha+\beta}\nabla^2f(\mZ)\right](\mG,\mH) - \left\langle \mG,\mH \right\rangle\right| &= \left|\left[\frac{2}{\alpha+\beta}\nabla^2f(\mZ) - \mId\right](\mG,\mH)\right| \\
%& \leq \frac{1}{\|\mG\|_F}\left|\left[\frac{2}{\alpha+\beta}\nabla^2f(\mZ) - \mId\right](\mG,\mG)\right| \left\|\mH\right\|_F \leq \frac{\beta - \alpha}{\beta + \alpha}\left\|\mG\right\|_F \left\|\mH\right\|_F.
%\end{align*}
%\end{proof}
%\begin{lem}\label{lem:bound WDelta}\cite{li2016}
%For any matrix $\mC,\mD \in\R^{n\times r}$, let $\mP_{\mC}$ be the orthogonal projector onto the range of $\mC$. Let $\mR = \argmin_{\widetilde \mR\in\calO_r}\|\mC - \mD\mR\|_F$. Then
% \begin{align}
% \left\|\mC\left(\mC - \mD\mR\right)^\T \right\|_F^2 \leq \frac{1}{8}\left\|\mC\mC^\T - \mD\mD^\T \right\|_F^2 + \left(3+\frac{1}{2(\sqrt{2} - 1)}\right) \left\|\left(\mC\mC^\T - \mD\mD^\T\right)\mP_{\mC} \right\|_F^2.
% \label{eq:bound WDelta}\end{align}
% \end{lem}
The following result provides an upper bound on the energy of the difference $\mW\mW^\T - \mW^\star\mW^{\star\T}$ when projected onto the column space of $\mW$. Its proof is given in Appendix~\ref{sec:proof bound WW - W*W*QQ}.
\begin{lem}\label{lem:bound:WW}
Suppose $f(\mX)$ satisfies the $(2r,4r)$-restricted strong convexity and smoothness condition \eqref{eq:RIP like}. For any critical point $\mW$ of \eqref{eq:factored problem}, let $\mP_{\mW}\in\R^{(m+n)\times (m+n)}$ be the orthogonal projector onto the column space of $\mW$. Then
\[
\left\|(\mW\mW^\T - \mW^\star\mW^{\star\T})\mP_{\mW} \right\|_F \leq 2\frac{\beta - \alpha}{\beta + \alpha} \left\|\mX -\mX^\star\right\|_F.
\]
 \end{lem}
We remark that Lemma~\ref{lem:bound:WW} is a variant of \cite[Lemma 3.2]{park2016non}. While the result there requires the $4r$-RIP condition of the objective function, our result depends on the $(2r,4r)$-restricted strong convexity and smoothness condition. Our result is also slightly tighter than \cite[Lemma 3.2]{park2016non}.

In addition, for any matrices $\mC,\mD \in\R^{n\times r}$, the following result relates the distance between $\mC\mC^\T$ and $ \mD\mD^\T$ to the distance between $\mC$ and $\mD$.

\begin{lem}\label{lem:CC - DD to WDelta} For any matrices $\mC,\mD \in\R^{n\times r}$ with ranks $r_1$ and $r_2$, respectively,  let $\mR = \argmin_{\mR'\in\calO_r}\|\mC - \mD\mR'\|_F$. Then
 \begin{align*}
 &\|\mC\mC^\T - \mD\mD^\T\|_F^2{\large/}\|\mC - \mD\mR\|_F^2 \\
 &\geq \max\left\{2(\sqrt{2}-1)\sigma_r^2(\mD),\min\left\{\sigma_{r_1}^2(\mC), \sigma_{r_2}^2(\mD)\right\}\right\}.
\end{align*}
If $\mC = \mzero$, then we have
 \begin{align*}
 &\left\|\mC\mC^\T - \mD\mD^\T \right\|_F^2\geq \sigma_{r_2}^2(\mD)\left\|\mC - \mD\mR\right\|_F^2.
\end{align*}
 \end{lem}

We present one more useful result in the following Lemma.
\begin{lem}\label{lem:bound WDelta}\cite[Lemma 3]{li2016}
For any matrices $\mC,\mD \in\R^{n\times r}$, let $\mP_{\mC}$ be the orthogonal projector onto the range of $\mC$. Let $\mR = \argmin_{\mR'\in\calO_r}\|\mC - \mD\mR'\|_F$. Then
 \begin{align*}
 &\|\mC\left(\mC - \mD\mR\right)^\T \|_F^2 \leq \frac{1}{8}\|\mC\mC^\T - \mD\mD^\T \|_F^2 \\
 &+ (3+\frac{1}{2(\sqrt{2} - 1)}) \|(\mC\mC^\T - \mD\mD^\T)\mP_{\mC}\|_F^2.
\end{align*}
 \end{lem}
Finally, we provide the gradient and Hessian expressions for $\rho(\mW)$.  The gradient of $\rho(\mW)$ is given by
\begin{align*}
&\nabla_{\mU}\rho(\mU,\mV) = \nabla f(\mX)\mV + \mu\mU(\mU^\T\mU - \mV^\T\mV),\\
&\nabla_{\mV}\rho(\mU,\mV) = \nabla f(\mX)^\T\mU - \mu\mV(\mU^\T\mU - \mV^\T\mV).
\end{align*}
Standard computations give the the Hessian quadrature form  $[\nabla^2 \rho(\mW)](\mDelta,\mDelta)$ for any $\mDelta = \begin{bmatrix} \mDelta_{\mU}\\ \mDelta_{\mV}\end{bmatrix}$ where $\mDelta_{\mU}\in\R^{n\times r},\mDelta_{\mV}\in\R^{m\times r}$:
\begin{align*}
&[\nabla^2\rho(\mW)](\mDelta,\mDelta) \\
&= [\nabla^2f(\mX)](\mDelta_{\mU}\mV^\T+ \mU\mDelta_{\mV}^\T,\mDelta_{\mU}\mV^\T + \mU\mDelta_{\mV}^\T)\\
& \quad + 2\langle \nabla f(\mX),\mDelta_{\mU}\mDelta_{\mV}^\T \rangle+ [\nabla^2g(\mW)](\mDelta,\mDelta),
\end{align*}
where
\begin{align*}
&[\nabla^2g(\mW)](\mDelta,\mDelta)
= \mu\langle \widehat\mW^\T\mW,\widehat\mDelta^\T\mDelta\rangle\\
 &\quad+ \mu\langle \widehat\mW\widehat\mDelta^\T,\mDelta \mW^\T \rangle + \mu\langle \widehat\mW\widehat\mW^\T,\mDelta \mDelta^\T\rangle.
\end{align*}

\subsection{The Formal Proof}
\begin{proof}[Proof of Theorem~\ref{thm:stricit saddle}]
Any critical point $\mW$ of $\rho
(\mW)$ satisfies $\nabla \rho(\mW) = \mzero$, i.e.,
\begin{align}
&\nabla f(\mX)\mV + \mu\mU\left(\mU^\T\mU - \mV^\T\mV\right) = \mzero,\label{eq:proof thm cirtical 1}\\
& \nabla f(\mX)^\T\mU - \mu\mV\left(\mU^\T\mU - \mV^\T\mV\right) = \mzero.\label{eq:proof thm cirtical 2}
\end{align}
By~\eqref{eq:proof thm cirtical 2}, we obtain
\[
\mU^\T \nabla f(\mX) = \mu\left(\mU^\T\mU - \mV^\T\mV\right)\mV^\T.
\]
Multiplying \eqref{eq:proof thm cirtical 1} by $\mU^\T$ and plugging in the expression for $\mU^\T \nabla f(\mX)$ from the above equation $\mV^\T$ gives
\begin{align*}
(\mU^\T\mU - \mV^\T\mV)\mV^\T\mV + \mU^\T\mU(\mU^\T\mU - \mV^\T\mV) = \mzero,
\end{align*}
which further implies
\begin{align*}
\mU^\T\mU\mU^\T\mU=\mV^\T\mV\mV^\T\mV.
\end{align*}
Note that $\mU^\T\mU$ and $\mV^\T\mV$ are the principal square roots (i.e., PSD square roots) of $\mU^\T\mU\mU^\T\mU$ and $\mV^\T\mV\mV^\T\mV$, respectively. Utilizing the result that a PSD matrix has a unique principal square root~\cite{johnson2001uniqueness}, we obtain
\begin{align}
\mU^\T\mU = \mV^\T\mV.
\label{eq:critical point balanced}\end{align}
Thus, we can simplify \eqref{eq:proof thm cirtical 1} and \eqref{eq:proof thm cirtical 2} by
\begin{align}
&\nabla_{\mU}\rho(\mU,\mV) = \nabla f(\mX)\mV = \mzero,\label{eq:proof thm cirtical 3}\\
&\nabla_{\mV}\rho(\mU,\mV) = \nabla f(\mX)^\T\mU  = \mzero.\label{eq:proof thm cirtical 4}
\end{align}
Now we turn to prove the strict saddle property and that there are no spurious local minima.

First, note that as guaranteed by Proposition~\ref{prop:RIP to unique}, $\mX^\star$ is the unique $n\times m$ matrix with rank at most $r$. Also the gradient of $f(\mX)$ vanishes at $\mX^\star$ since~\eqref{eq:original problem} is an unconstraint optimization problem. Denote the set of critical points of $\rho(\mW)$ by
\begin{align*}
\calC:=\left\{\mW\in\R^{(n+m)\times r}:\nabla \rho(\mW)=\mzero\right\}.
\end{align*}
We separate $\calC$ into two subsets:
\begin{align*}
\calC_1: &= \calC\cap \left\{\mW\in\R^{(n+m)\times r}: \mU\mV^\T= \mX^\star\right\},\\
\calC_2: &= \calC\cap \left\{\mW\in\R^{(n+m)\times r}: \mU\mV^\T\neq \mX^\star\right\},
\end{align*}
satisfying $\calC = \calC_1 \cup \calC_2$. Since any critical point $\mW$ satisfies \eqref{eq:critical point balanced}, $g(\mW)$ achieves its global minimum at $\mW$. Also $f(\mX)$ achieves its global minimum at $\mX^\star$. We conclude that $\mW$ is the globally optimal solution of $\rho$ for any $\mW\in\calC_1$. If we show that any $\mW\in\calC_2$ is a strict saddle, then we prove that there are no spurious local minima as well as the strict saddle property. Thus, the remaining part is to show that $\calC_2$ is the set of strict saddles.

To show that $\calC_2$ is the set of strict saddles, it is sufficient to find a direction $\mDelta$ along which the Hessian has a strictly negative curvature for each of these points. %Following the same approach as in~\cite{bhojanapalli2016lowrankrecoveryl,park2016non,li2016},
We construct $\mDelta = \mW - \mW^\star\mR$, the difference from $\mW$ to its nearest global factor $\mW^\star$, where
\begin{align*}
\mR = \argmin_{\mR'\in\calO_r}\left\|\mW - \mW^\star\mR'\right\|_F.
\end{align*}
Such $\mDelta$ satisfies $\mDelta \neq \mzero$ since $\mX \neq \mX^\star$ implying $\mW\mW^\T\neq \mW^{\star}\mW^{\star\T}$.
Then we evaluate the Hessian bilinear form along the direction $\mDelta$:
\begin{equation}\begin{split}
&[\nabla^2\rho(\mW)](\mDelta,\mDelta) = 2\underbrace{\langle \nabla f(\mX),\mDelta_{\mU}\mDelta_{\mV}^\T \rangle}_{\Pi_1} \\
&+\underbrace{[\nabla^2f(\mX)](\mDelta_{\mU}\mV^\T+\mU\mDelta_{\mV}^\T,\mDelta_{\mU}\mV^\T+\mU\mDelta_{\mV}^\T)}_{\Pi_2}\\
&+ \mu \underbrace{\langle \widehat\mW\widehat\mDelta^\T,\mDelta \mW^\T  \rangle}_{\Pi_3} +\mu \underbrace{\langle \widehat\mW\widehat\mW^\T,\mDelta \mDelta^\T\rangle}_{\Pi_4}.
\end{split}
\label{eq:Hessian for Delta}
\end{equation}
The following result (which is proved in Appendix~\ref{sec:proof eq bound Pi}) states that $\Pi_1$ is strictly negative, while the remaining terms are relatively small, though they may be nonnegative:
\begin{equation}\begin{split}
\Pi_1 &\leq- \alpha \left\| \mX -\mX^\star \right\|_F^2,\quad \Pi_2 \leq  \beta \|\mW\mDelta^\T\|_F^2,\\
\Pi_3 &\leq  \|\mW\mDelta^\T\|_F^2,\quad \Pi_4 \leq 2 \left\| \mX -\mX^\star \right\|_F^2.\\
\end{split}
\label{eq:bound Pi}\end{equation}
Now, substituting \eqref{eq:bound Pi} into \eqref{eq:Hessian for Delta} gives
\begin{equation}\begin{split}
&[\nabla^2\rho(\mW)](\mDelta,\mDelta) \\
&= 2\Pi_1+ \Pi_2  + \mu \Pi_3 +\mu \Pi_4 \\
& \leq - 2\alpha \|\mX - \mX^\star\|_F^2 + (\beta+\mu) \cdot\|\mW\mDelta^\T\|_F^2 \\
&\quad + 2\mu\left\|\mX - \mX^\star\right\|_F^2 \\
&\stackrel{(i)}{\leq}(-2\alpha + 2\mu) \left\|\mX -\mX^\star\right\|_F^2\\
&\quad + (\beta+\mu)(\frac{1}{2}+ (12+\frac{2}{\sqrt{2} - 1})(\frac{\beta - \alpha}{\beta + \alpha})^2) \left\|\mX -\mX^\star\right\|_F^2\\
&\stackrel{(ii)}{\leq} -0.2\alpha \left\|\mX -\mX^\star\right\|_F^2,
\end{split}
\label{eq:Hessian is negative}\end{equation}
where $(i)$ utilizes Lemmas \ref{lem:bound:WW} and \ref{lem:bound WDelta}, $(ii)$ utilizes the following inequality (which is proved in Appendix~\ref{sec:prf eq bound 5})
%(which is proved in Appendix~\ref{sec:proof eq bound 5})
\begin{align}
\left\|\mW\mW^\T - \mW^\star\mW^\star\right\|_F^2 \leq 4\left\|\mX - \mX^\star\right\|_F^2,
\label{eq:bound 5}\end{align}
and $(ii)$ holds because $\frac{\beta}{\alpha} \leq 1.5$ and $\mu \leq \frac{1}{16}\alpha$. Thus, if $\mX\neq \mX^\star$, $\left[\nabla^2\rho(\mX)\right](\mDelta,\mDelta)$ is always negative. This implies that $\mW$ is a strict saddle.

To complete the proof, we utilize Lemma~\ref{lem:CC - DD to WDelta}  to further bound the last term in~\eqref{eq:Hessian is negative}:
 \begin{align*}
&[\nabla^2\rho(\mW)](\mDelta,\mDelta)  \leq -0.05 \alpha \|\mW\mW^\T -\mW^\star{\mW^\star}^\T\|_F^2\\
&\leq-0.05\alpha\|\mDelta\|_F^2\left\{\begin{matrix} 2(\sqrt{2}-1)\sigma_{r}^2(\mW^\star), & r = r^\star, \\ \min\left\{\sigma_{r^c}^2(\mW),\sigma_{r^\star}^2(\mW^\star)\right\}, & r>r^\star,\\
 \sigma_{r^\star}^2(\mW^\star), & r_c = 0,
\end{matrix}\right.
\end{align*}
where $r^c$ is the rank of $\mW$, the fist inequality utilizes~\eqref{eq:bound 5}, and
the second inequality follows from Lemma \ref{lem:CC - DD to WDelta}. We complete the proof of Theorem~\ref{thm:stricit saddle} by noting that $\sigma_{\ell}^2(\mW^\star) = 2\sigma_\ell(\mX^\star)$ for all $\ell\in\{1,\ldots, r^\star\}$ since
\[
\mW^\star = \begin{bmatrix} \mQ_{\mU^\star}{\mSigma^\star}^{1/2} \\ \mQ_{\mV^\star}{\mSigma^\star}^{1/2} \end{bmatrix} = \begin{bmatrix} \mQ_{\mU^\star}/\sqrt{2} \\ \mQ_{\mV^\star}/\sqrt{2} \end{bmatrix} \left(\sqrt{2}{\mSigma^\star}^{1/2}\right)\mId
\]
 is an SVD of $\mW^\star$, where we recall that $\mX^\star = \mQ_{\mU^\star}\mSigma^\star\mQ_{\mV^\star}^\T$ is an SVD of $\mX^\star$.
 \end{proof}
%by noting that
%$$\mW^\star = \begin{bmatrix} \mQ_{\mU^\star}{\mSigma^\star}^{1/2} \\ \mQ_{\mV^\star}{\mSigma^\star}^{1/2} \end{bmatrix} = \begin{bmatrix} \mQ_{\mU^\star}/\sqrt{2} \\ \mQ_{\mV^\star}/\sqrt{2} \end{bmatrix} \left(\sqrt{2}{\mSigma^\star}^{1/2}\right)\mId$$
% is an SVD of $\mW^\star$, where we recall that $\mX^\star = \mQ_{\mU^\star}\mSigma^\star\mQ_{\mV^\star}^\T$ is an SVD of $\mX^\star$.
\begin{remark}
From \eqref{eq:Hessian is negative}, we observe that a smaller $\mu$ yields a more negative bound on $\left[\nabla^2\rho(\mX)\right](\mDelta,\mDelta)$. This can be explained intuitively as follows. First note that any critical point $\mW$ satisfies \eqref{eq:critical point balanced} provided $\mu>0$, no matter how large or small $\mu$ is. The Hessian information about $g(\mW)$ is represented by the terms $\Pi_3$ and $\Pi_4$. We have
\begin{align*}
\Pi_3 + \Pi_4 &= \left\langle \widehat\mW\widehat\mDelta^\T,\mDelta \mW^\T  \right\rangle + \left\langle \widehat\mW\widehat\mW^\T,\mDelta \mDelta^\T\right\rangle\\
& = \left\langle \widehat\mW^\T\mDelta,\mDelta^\T\widehat\mW\right\rangle + \left\langle \widehat\mW^\T\mDelta,\widehat\mW^\T\mDelta\right\rangle\\
& = \left\langle \widehat\mW^\T\mDelta,\widehat\mW^\T\mDelta + \mDelta^\T\widehat\mW\right\rangle\\
& \geq 0,
\end{align*}
where the last line holds since for any $r\times r$ matrix $\mA$,
\begin{align*}
&\left\langle \mA, \mA + \mA^\T\right\rangle\\ &= \frac{1}{2}\left\langle \mA + \mA^\T, \mA + \mA^\T\right\rangle +  \frac{1}{2}\left\langle \mA - \mA^\T, \mA + \mA^\T\right\rangle \\ &= \frac{1}{2}\left\| \mA + \mA^\T\right\|_F^2\geq 0.
\end{align*}
%\[
%\left\langle \mA, \mA + \mA^\T\right\rangle = \frac{1}{2}\left\langle \mA + \mA^\T, \mA + \mA^\T\right\rangle +  \frac{1}{2}\left\langle \mA - \mA^\T, \mA + \mA^\T\right\rangle = \frac{1}{2}\left\| \mA + \mA^\T\right\|_F^2\geq 0.
%\]
Thus the Hessian of $\rho$ evaluated at any critical point $\mW$ is a PSD matrix\footnote{This can also be observed since any critical point $\mW$ is a global minimum point of $\rho(\mW)$, which directly indicates that $\nabla^2\rho(\mW)\succeq \mzero$.} instead of having a negative eigenvalue. In low-rank, PSD matrix optimization problems, the corresponding objective function (without any regularizer {such as} $g(\mW)$) is proved to have the strict saddle property~\cite{bhojanapalli2016lowrankrecoveryl,li2016}. Therefore, $h(\mW)$ is also expected to have the strict saddle property, and so is $\rho(\mW)$ when $\mu$ is small, i.e., the Hessian of $g(\mW)$ has little influence on the Hessian of $\rho(\mW)$ when $\mu$ is small. Our results also indicate that when the restricted strict convexity constant $\alpha$ is not provided a priori, we can always choose a small $\mu$ to ensure the strict saddle property of $\rho(\mW)$ is met, and hence we are guaranteed the global convergence of a number of local search algorithms applied to~\eqref{eq:factored problem}.
\end{remark}

\section{Experiments}\label{sec:experiments}
In this section, we present a set of experiments on matrix sensing, matrix completion, and 1-bit matrix completion to demonstrate the performance of iterative algorithms for low-rank matrix optimization. Unless noted otherwise, we denote the matrix factorization approach by NVX and
use the minFunc package\footnote{Software available at\\ https://www.cs.ubc.ca/$\sim$schmidtm/Software/minFunc.html} to perform the local search algorithms for the factored problem.

\subsection{Matrix Sensing}\label{sec:exp matrix sensing}
We first present some experiments to illustrate the performance of local search algorithms for the matrix sensing problem with the factorization approach~\eqref{eq:matrix sensing}. In these experiments, we set $n = 50$, $m =50$ and vary the rank $r$ from $1$ to $19$. We generate a rank-$r$ $n\times m$ random matrix $\mX^\star$ by setting $\mX^\star = \widetilde\mU\widetilde\mV^\T$ where $\widetilde\mU$ and $\widetilde\mV$ are respectively $n\times r$ and $m\times r$ matrixes of normally distributed random numbers. We then obtain $p$ random measurements $\vy = \calA(\mX^\star)$ with
\[
y_i = \left\langle \mX^\star, \mY_i \right\rangle,
\]
where the entries of each $n\times m$ matrix $\mY_i$ are independent and identically distributed (i.i.d.) normal random variables with zero mean and variance $\frac{1}{p}$ for $i\in\{1,2,\ldots p\}$. For each pair of $r$ and the number of measurements, 10 Monte Carlo trials are carried out and for each trial, {and} we claim matrix recovery to be successful if the relative reconstruction error satisfies
\[
\frac{\|\mX^\star - \widehat \mX\|_F}{\|\mX^\star\|_F}\leq 10^{-4},
\]
where we denote by $\widehat \mX$ the reconstructed matrix. Figure~\ref{fig:Phase} displays the phase transition for factorized gradient descent starting from a random initialization, the singular value projection (SVP) method proposed in~\cite{jain2010guaranteed} which requires a SVD in each iteration, and the convex approach which solves
\begin{equation}\begin{split}
&\minimize_{\mX} \|\mX\|_* \\
& \st \vy = \calA(\mX).
\end{split}\label{eq:nuclear norm for sensing}\end{equation}
{We see that there are only negligible differences between the different approaches for matrix sensing; these approaches also have very similar performance guarantees when the Gaussian sensing operator $\calA$ satisfies the RIP~\cite{candes2011tight}.} We note that with or without the regularizer $g$ as defined in \eqref{eq:define g}, local search algorithms have similar performance with random initialization. Hence, {throughout all of} the experiments, we simply discard the regularizer $g$, but we stress that identical performance is observed if we have this regularizer $g$.

{The previous experiments suppose that $r$ is known for  SVP and the matrix factorization approach.  We note, however, that our result in Theorem~\ref{thm:stricit saddle} also covers the over-parameterization case where $r>r^\star$. To illustrate the possible influence of over-parameterization, we generate a rank-$r^\star$ random matrix $\mX^\star\in\R^{n\times m}$ with $r^\star = 4$ and $n = m = 50$ and obtain $p = 4Rn$ random measurements (so that the measurement operator $\calA$ satisfies the RIP of rank $R$), where $R = 7$. We then solve the matrix factorization problem\footnote{To avoid tuning the parameters (such as step-size) for different $r$, we use the minFunc package with the default setting, which solves the factored problem by the ``LBFGS" algorithm~\cite{liu1989limited}.} with $r = 4,5,6,7$ and display the corresponding convergence results in Figure~\ref{fig:over param}. As can been seen, the matrix factorization approach converges to the target matrix $\mX^\star$ in both the exact-parameterization and over-parameterization cases. However, we also observe that it converges slower in the over-parameterization case (i.e., $r>r^\star$) than in the exact-parameterization case (i.e., $r = r^\star$). %This can partly be explained by \eqref{eq:strict saddle} which states that the Hessian has much
}

\begin{figure}[htb!]
\begin{minipage}{0.52\linewidth}
\centerline{
\includegraphics[width=2in]{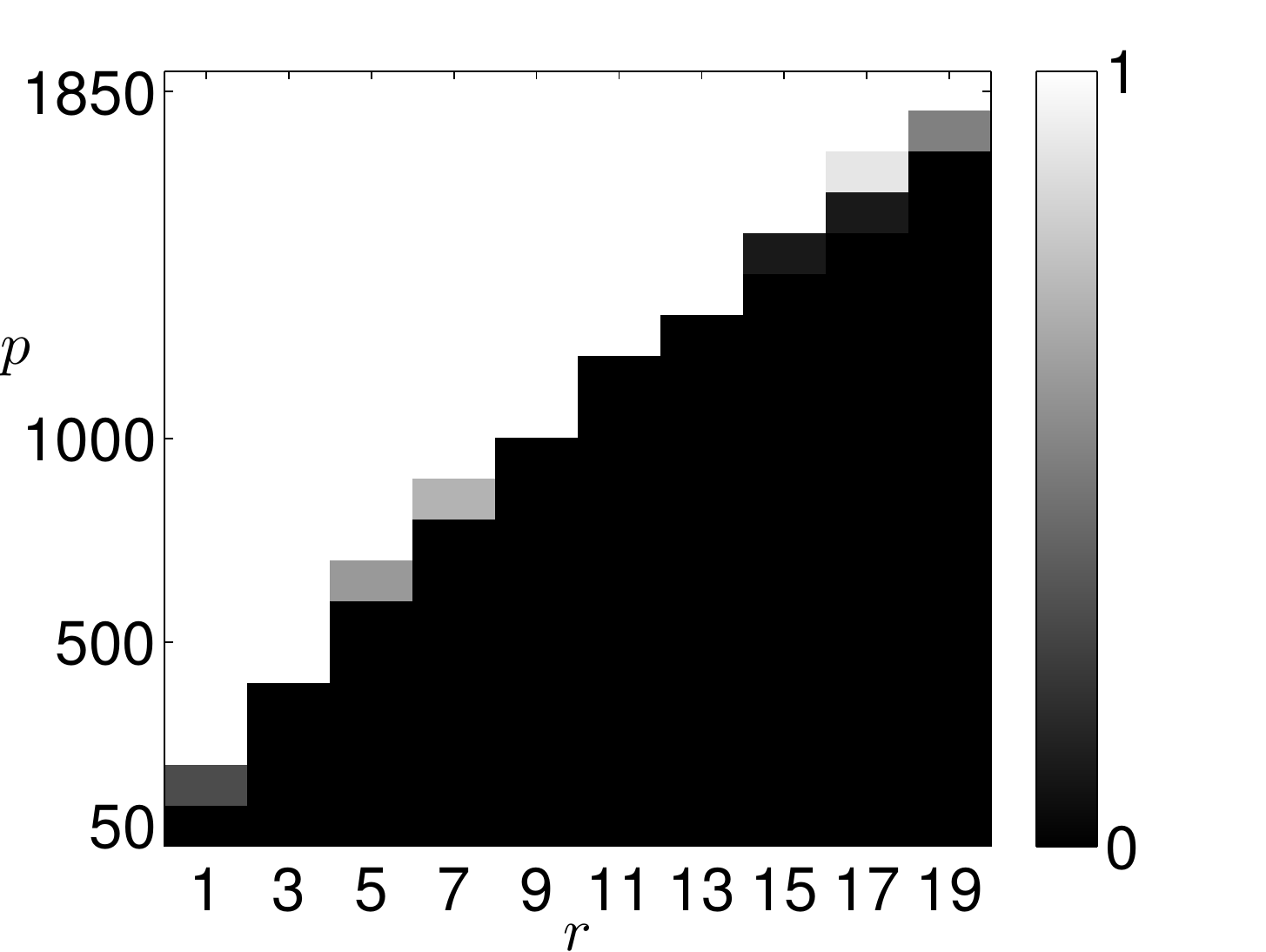}}
\centering{(a)}
\end{minipage}
\hfill
\begin{minipage}{0.44\linewidth}
\centerline{
\includegraphics[width=2in]{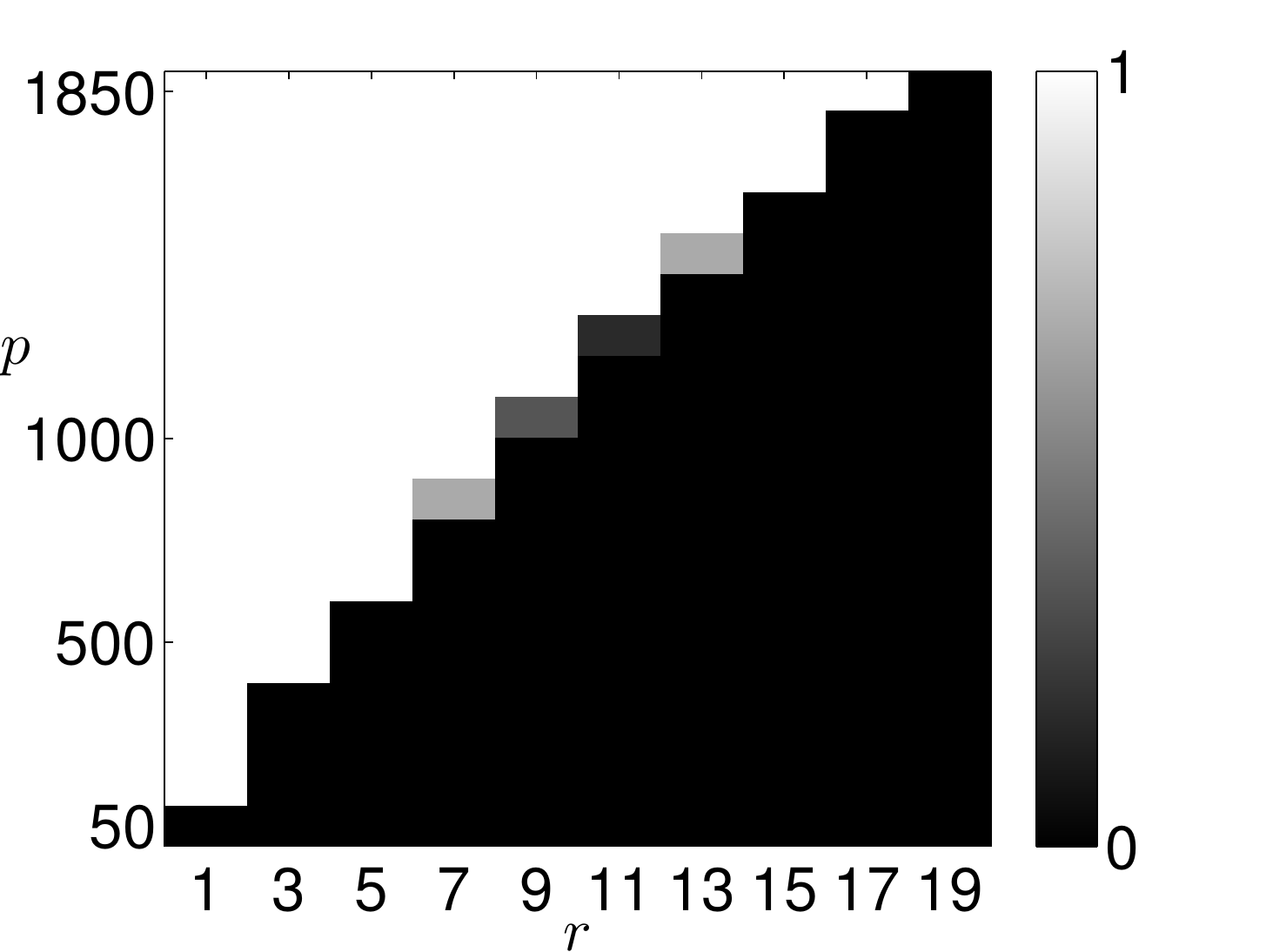}}
\centering{(b)}
\end{minipage}
\vfill
\centerline{
\includegraphics[width=2in]{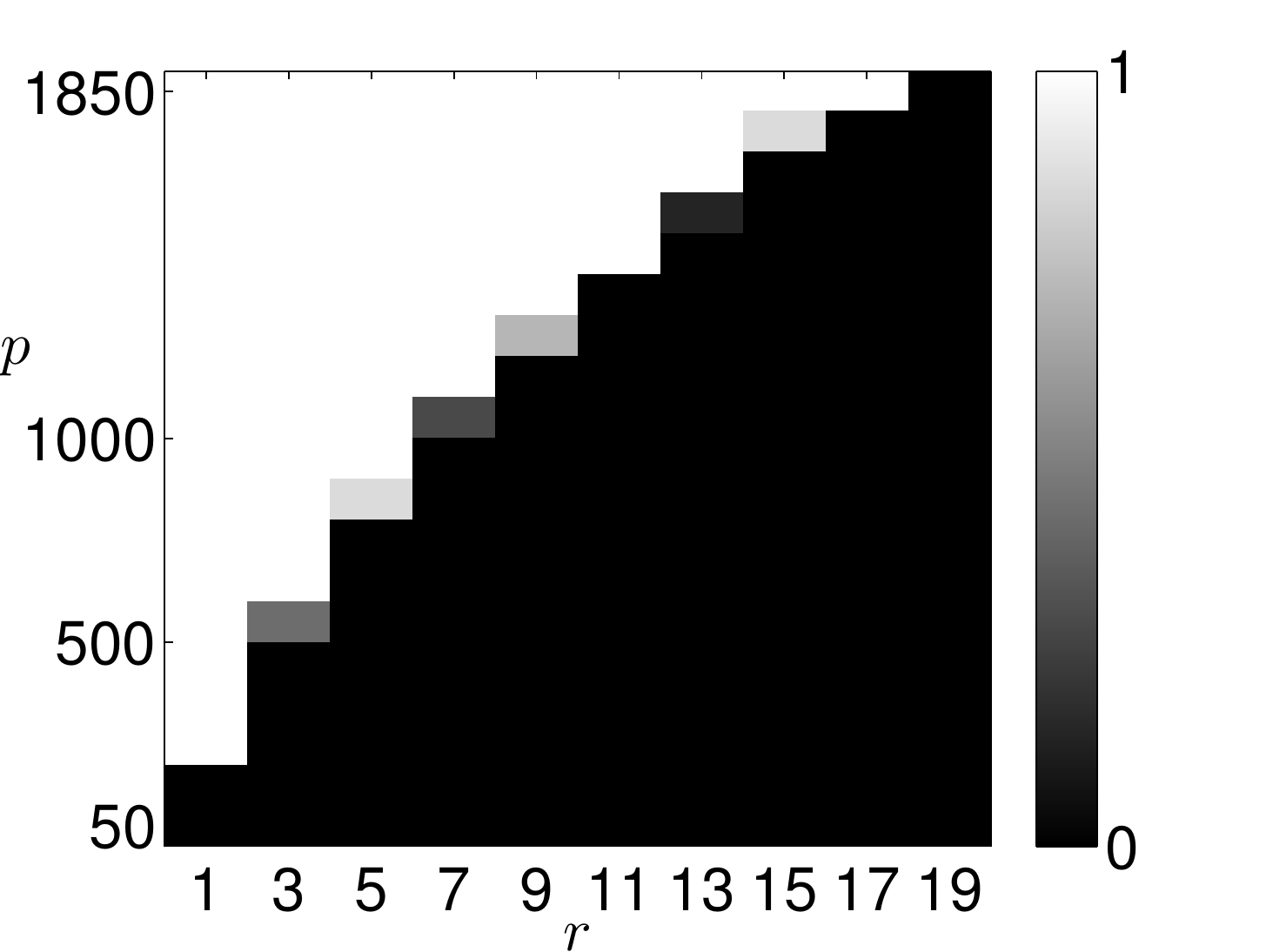}}
\centering{(c)}
\caption{\label{fig:Phase} Rate of success for matrix sensing by (a) solving the factorized problem~\eqref{eq:matrix sensing} with gradient descent; (b) SVP~\cite{jain2010guaranteed}; (c) solving the convex problem \eqref{eq:nuclear norm for sensing}.}
\end{figure}

\begin{figure}[htb!]
\begin{minipage}{0.48\linewidth}
\centerline{
\includegraphics[width=1.9in]{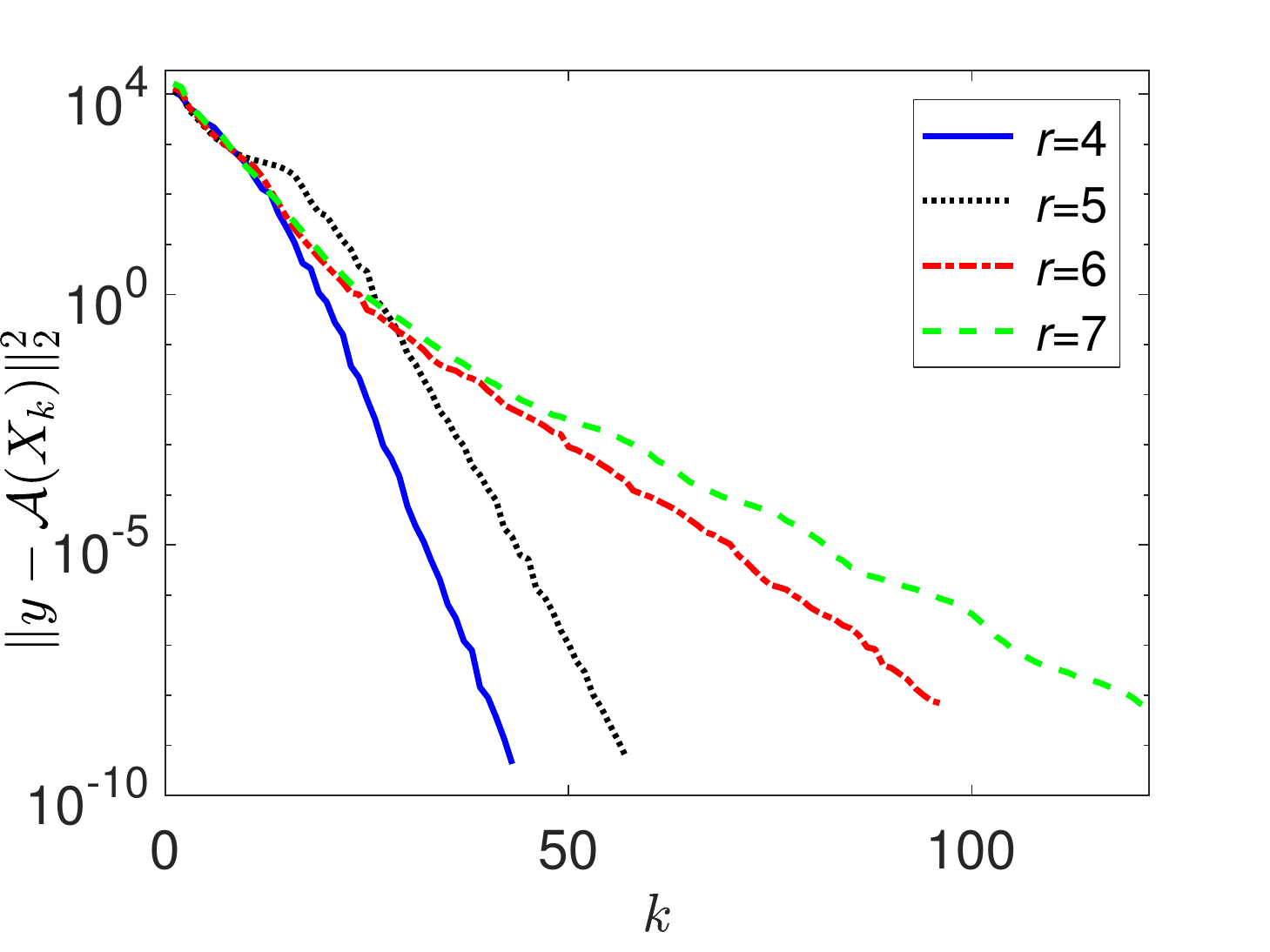}}
\centering{(a)}
\end{minipage}
\hfill
\begin{minipage}{0.48\linewidth}
\centerline{
\includegraphics[width=1.9in]{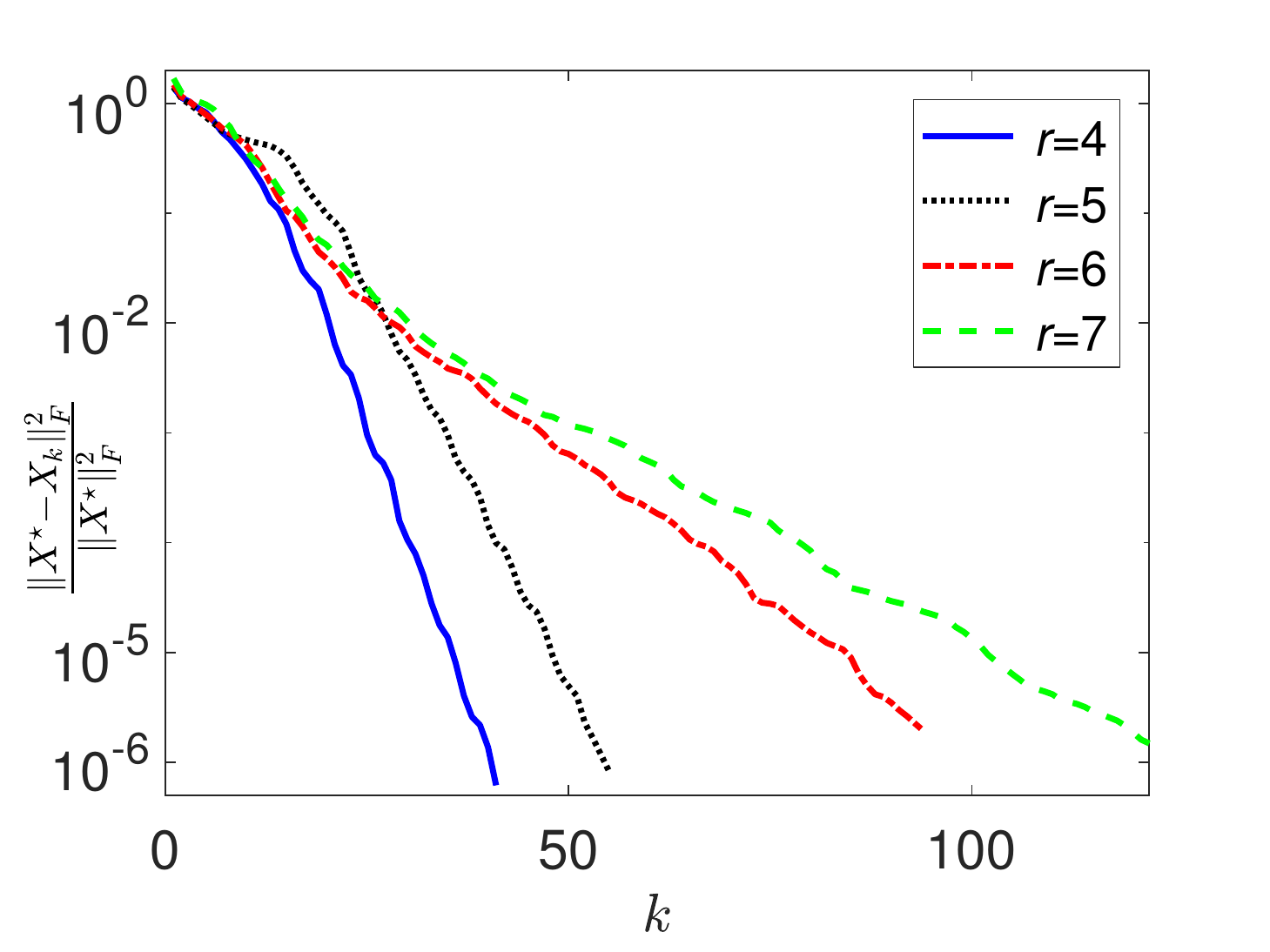}}
\centering{(b)}
\end{minipage}
\caption{\label{fig:over param}{ The performance in terms of (a) objective value and (b) the relative Frobenius norm of the error versus the iteration $k$ for the matrix factorization approach solving matrix sensing  with $r^\star = 4, n = m = 50, p = 4Rn, R = 7$ and $r$ varying from $r^\star$ to $R$.}}
\end{figure}

\subsection{Matrix Completion}
We compare the performance of the matrix factorization approach with SVP~\cite{jain2010guaranteed}, the convex approach, and singular value thresholding\footnote{Software available at http://svt.stanford.edu/} (SVT)~\cite{cai2010singular} for matrix completion where we want to recover a low-rank matrix $\mX^\star$ from incomplete measurements $\{X^\star_{ij}\}_{(i,j)\in\Omega}$, where $\Omega \subset [n]\times [m]$. Let $\calP_\Omega$ denote the projection onto the index set $\Omega$. The convex approach (denoted by CVX) attempts to use the nuclear norm as a convex relaxation of the rankness and solves
\begin{equation}\begin{split}
&\minimize_{\mX} \|\mX\|_* \\
& \st  \calP_\Omega(\mX) = \calP_\Omega(\mX^\star).
\end{split}\label{eq:nuclear norm for mc}\end{equation}
To make the recovery of $\mX^\star$ well-posed, we require $\mX^\star$ to be incoherent such that the information in $\mX$ is not concentrated in a small number of entries~\cite{candes2009exact}. A matrix $\mX\in\R^{n\times m}$ with singular value decomposition $\mX = \mL\mSigma \mQ^\T$ is $u$-incoherent if~\cite[Definition 2.1]{jain2010guaranteed}
\[
\max_{ij}|L_{ij}|\leq \sqrt\frac{u}{n},\ \max_{ij}|Q_{ij}|\leq \sqrt\frac{u}{m}.
\]
Though $\calP_\Omega$ does not satisfy the $r$-RIP \eqref{eq:RIP} for all low-rank matrices $\mX$, it satisfies the RIP when restricted to low-rank incoherent matrices.
\begin{thm} \cite[Theorem 4.2]{jain2010guaranteed} Without loss of generality, assume $n\geq m$. There exists a constant $C\geq 0$ such that for $\Omega\in[n]\times [m]$ chosen according to the Bernouli model with density greater than $C u^2 r^2 \log n /\delta^2 m$, with probability at least $1 - e^{-n\log n}$, the RIP holds for all $\mu$-incoherent matrices $\mX$ of rank at most $r$.
\label{thm:RIP MC}\end{thm}
Thus, {if local search algorithms (such as gradient descent) start with a random initialization and the iterates remain incoherent,} then Theorem~\ref{thm:stricit saddle} guarantees the global convergence of the matrix factorization approach with these algorithms. We note that this hypothesis is also required for SVP~\cite{jain2010guaranteed}. Though we can add a regularizer for incoherence as in \cite{ge2016matrix}, empirical evidence supports this hypothesis that the iterates in gradient descent are incoherent.

In the first set of experiments, we set $n=m =100$ and vary the rank $r$ from $1$ to $30$. Similar to the setup for matrix sensing in Section~\ref{sec:exp matrix sensing}, we generate a rank-$r$ random matrix and randomly obtain $p$ entries, i.e., $|\Omega| = p$. Figure~\ref{fig:Phase MC} displays the phase transition for gradient descent with a random initialization, SVP~\cite{jain2010guaranteed}, singular value thresholding (SVT)~\cite{cai2010singular}, and the convex approach. As can been seen, the matrix factorization approach has similar phase transition to SVP, and is slightly better than SVT and the convex approach in terms of the number of measurements needed for successful recovery.

In the second set of experiments, we set $r = 5$ and $p = 3r(2n - r)$ (3 times the number of degrees of freedom within a rank-$r$ $n\times n$ matrix), and vary $n$ from $40$ to $5120$. We compare the time needed for the four approaches in Figure~\ref{fig:time}; our matrix factorization approach is much faster than the other methods. The time savings for the matrix factorization approach comes from  avoiding performing the SVD, which is needed both for SVT and SVP in each iteration. We also observe that convex approach has {the} highest computational complexity and is not scalable (which is the reason that we only present its time for $n$ up to $640$).

\begin{figure}[htb!]
\begin{minipage}{0.48\linewidth}
\centerline{
\includegraphics[width=2in]{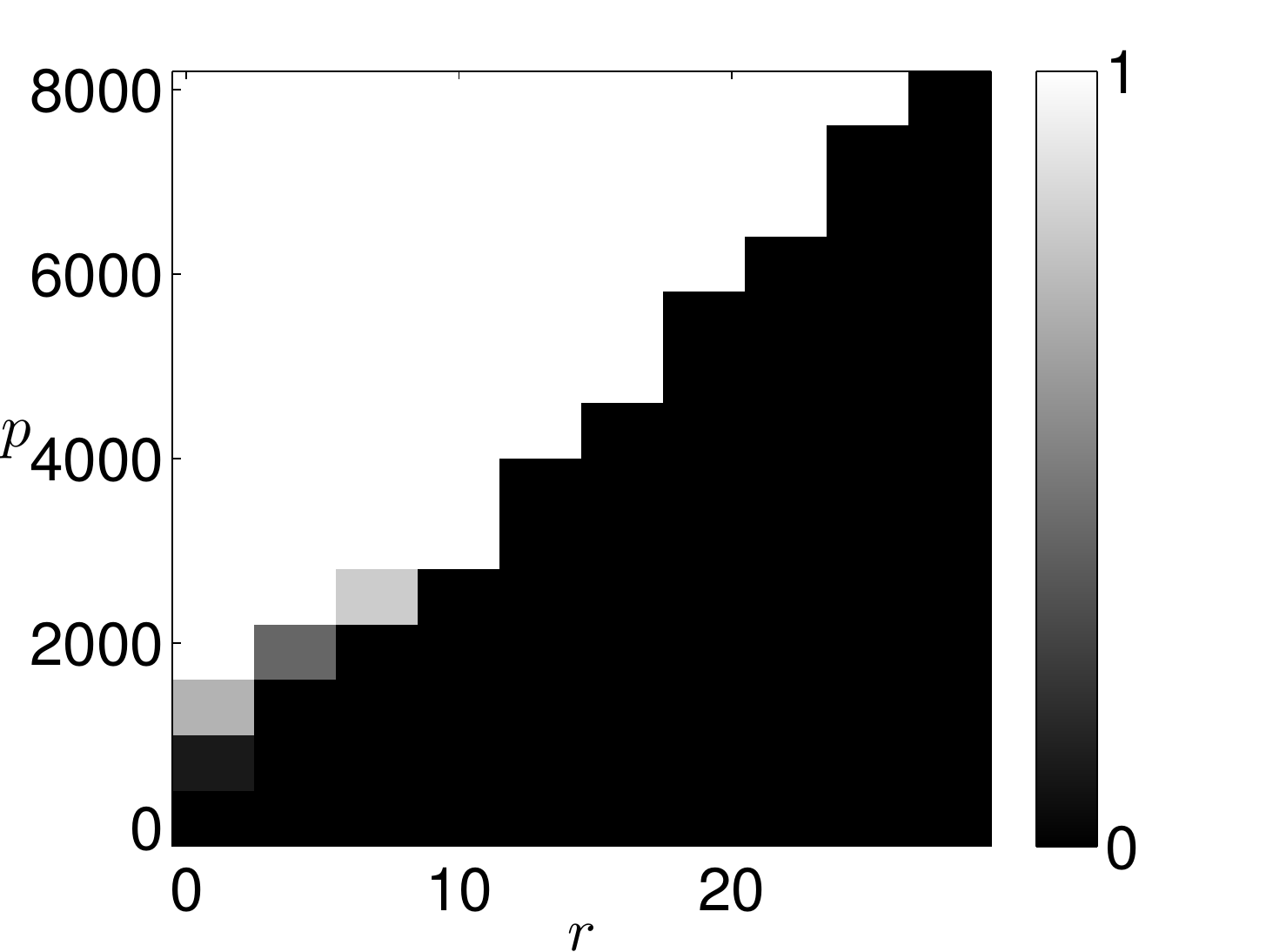}}
\centering{(a)}
\end{minipage}
\hfill
\begin{minipage}{0.48\linewidth}
\centerline{
\includegraphics[width=2in]{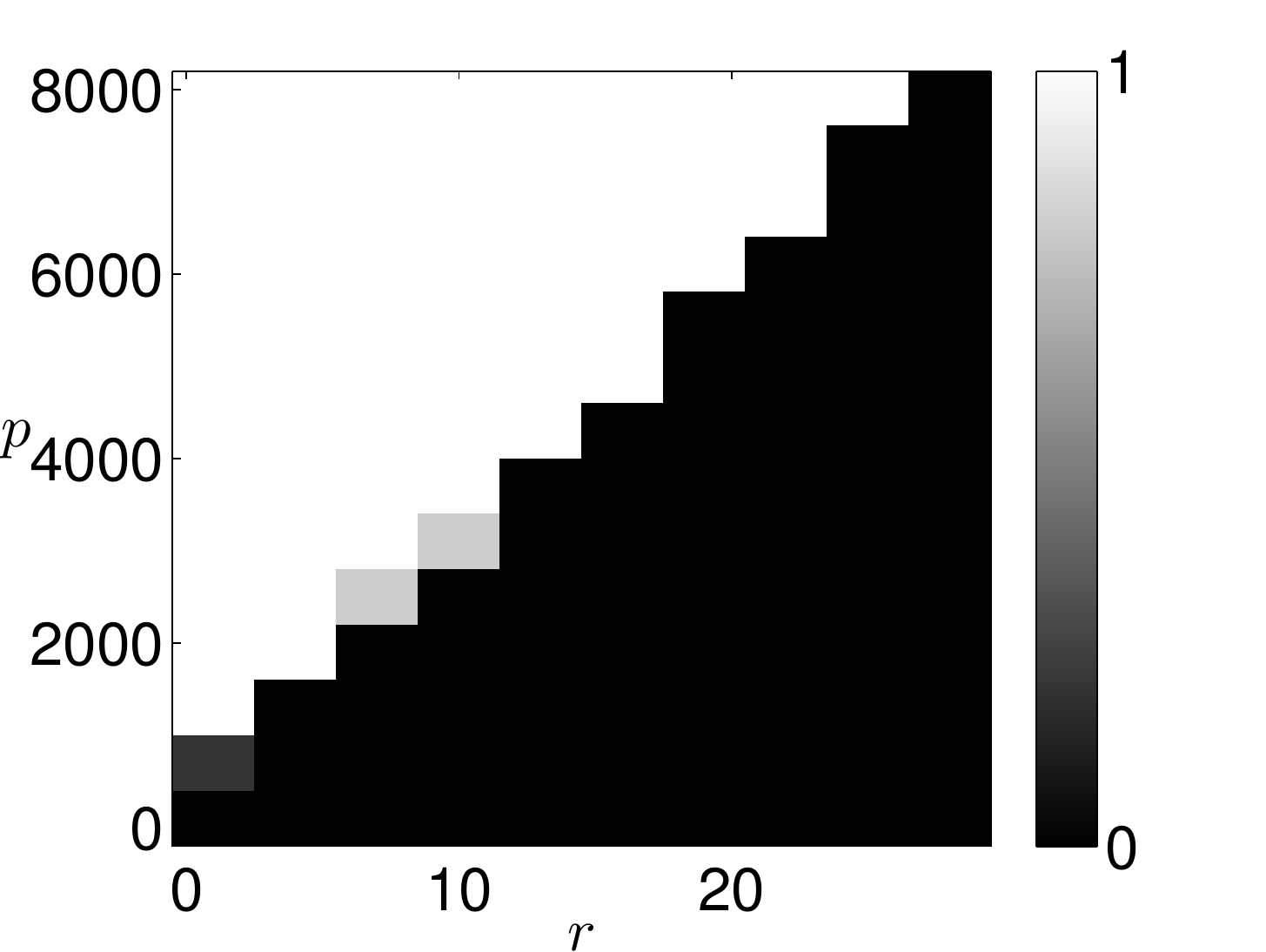}}
\centering{(b)}
\end{minipage}
\vfill
\begin{minipage}{0.48\linewidth}
\centerline{
\includegraphics[width=2in]{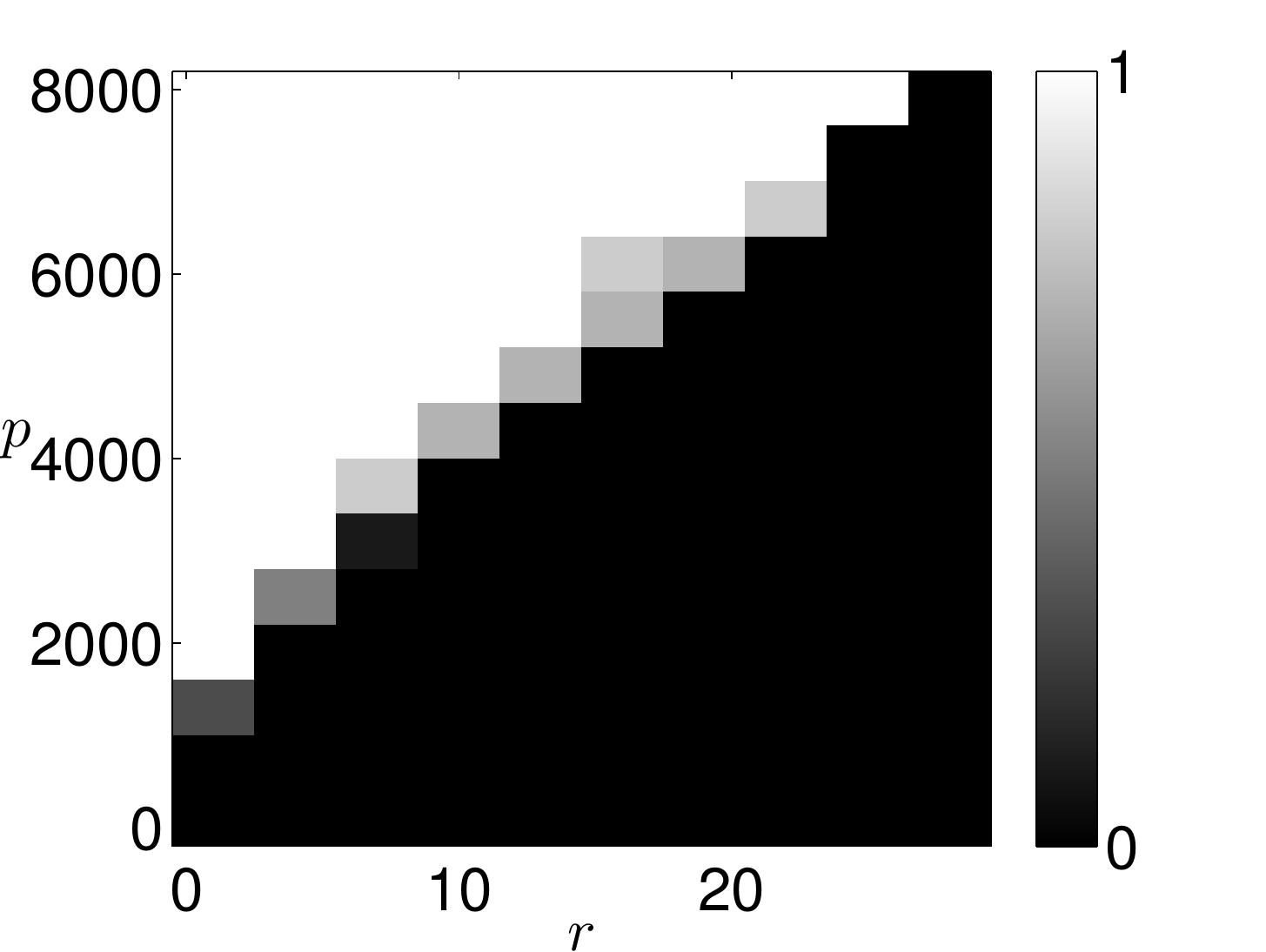}}
\centering{(c)}
\end{minipage}
\hfill
\begin{minipage}{0.48\linewidth}
\centerline{
\includegraphics[width=2in]{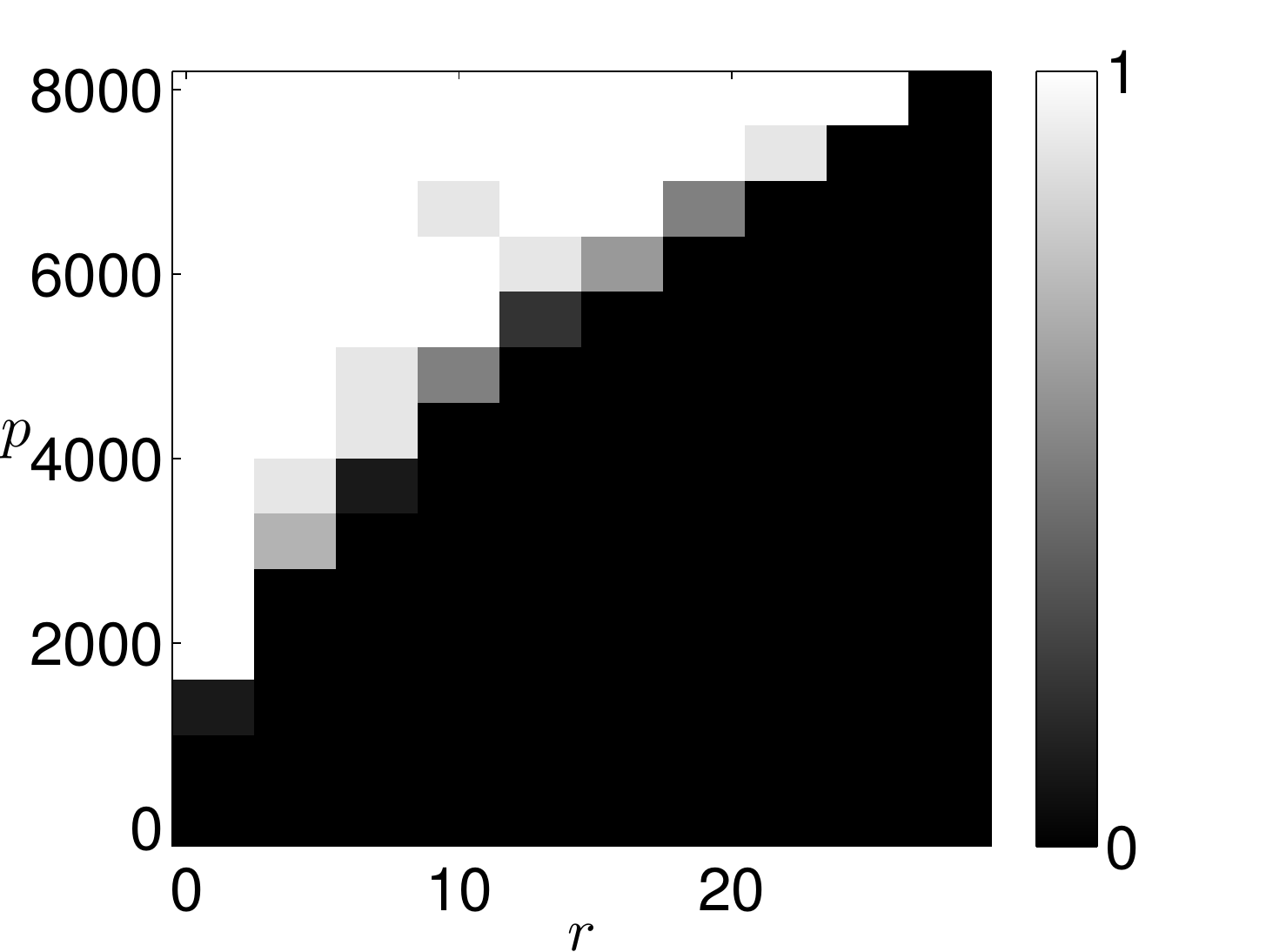}}
\centering{(d)}
\end{minipage}
\caption{\label{fig:Phase MC} Rate of success for matrix sensing by (a) the matrix factorization approach with gradient descent; (b) SVP~\cite{jain2010guaranteed}; (c) solving the convex problem \eqref{eq:nuclear norm for mc}; (d) SVT~\cite{jain2010guaranteed}. }
\end{figure}

\begin{figure}[htb!]
\includegraphics[width=3.5in]{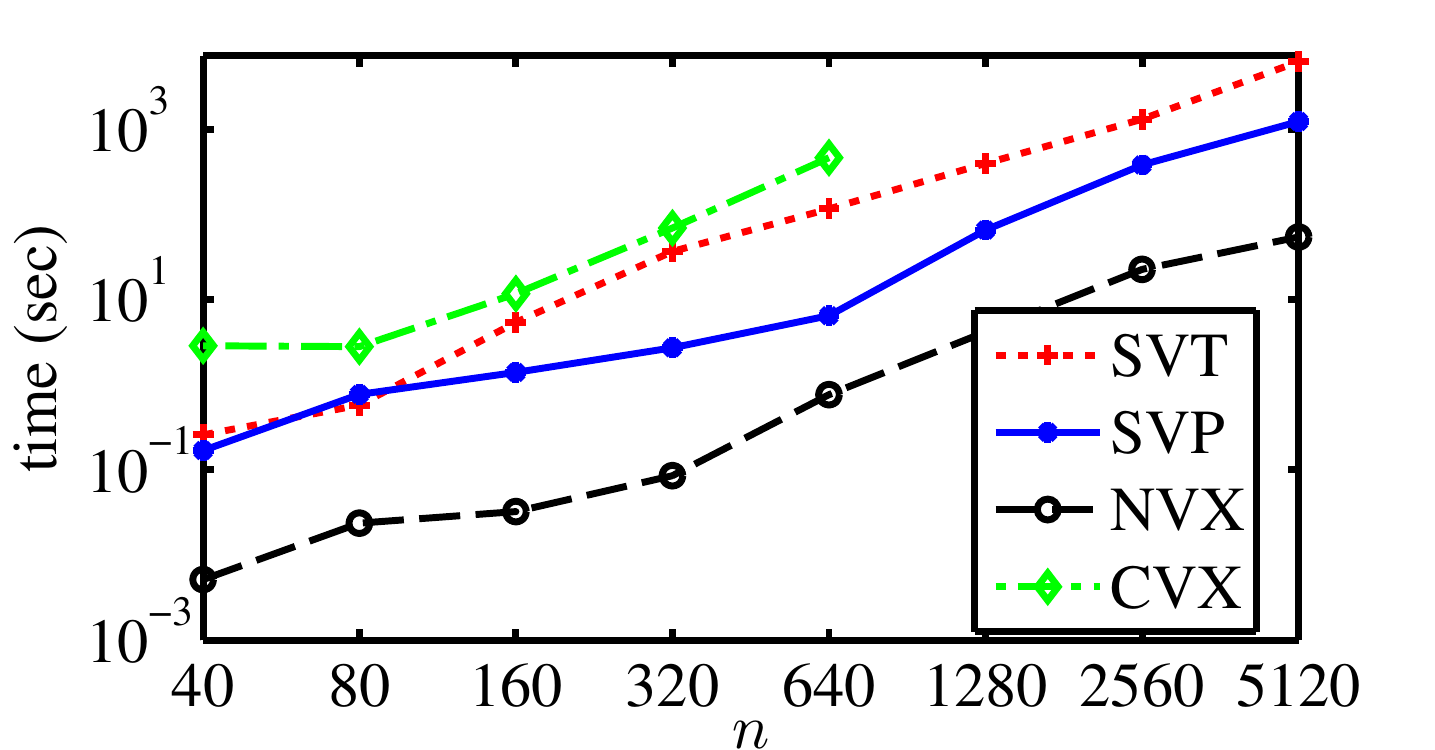}
\caption{\label{fig:time} Average computation time needed for different algorithms solving matrix completion.}
\end{figure}

\subsection{1-bit Matrix Completion}\label{sec:exp 1bit}
In the last set of experiments, we compare the performance of the matrix factorization approach with the convex approach\footnote{Software available at http://mdav.ece.gatech.edu/software/}  in \cite{davenport20141} for 1-bit matrix completion. We first note that to make the recovery problem well-posed, a constraint on $\|\mX\|_\infty$ (the entry-wise maximum of the matrix $\mX$) is applied in \cite{davenport20141} to require that the matrix is not too ``spiky''.
Instead of using the constraint on $\|\mX\|_\infty$, we add a smooth regularizer $\|\mX\|_F^2$ and turn to minimize the following objective function
\[
f_{\Omega,\mY}(\mX) = F_{\Omega,\mY}(\mX) + \frac{\eta}{2}\|\mX\|_F^2,
\]
which is also a convex function over $\mX$ and satisfies a similar restricted strong convexity and smoothness condition to $F_{\Omega,\mY}$ in Lemma~\ref{lem:1bit}. {In the case where we only observe part of the entries, then in light of Theorem \ref{thm:RIP MC}, the corresponding objective function is expected to satisfy the strong convexity and smoothness condition for all incoherent matrices.}
%Figure \ref{fig:1bit landscape}(b) shows the objective function $f_{\Omega,\mY}(\mX)$ for $2\times 2$ rank-1 PSD  matrix with $\eta =1$, the 1-bit measurements $\mY$ in \eqref{eq:1bit Y} and the probit regression model. As can been seen, with the additional regularizer $\|\mX\|_F^2$, the objective function $f_{\Omega,\mY}(\mX)$ behaves different than $F_{\Omega,Y}(\mX)$ in the left and right corners. In particular, unlike $F_{\Omega,\mY}$ which achieves its global minimum at $x_1\rightarrow \infty, x_2 \rightarrow -\infty$ or $x_1\rightarrow -\infty, x_2 \rightarrow \infty$, $f_{\Omega,\mY}(\mX)$ achieves its global minimum at $|x_1|\leq 1$ and $|x_2|\leq 1$ over all $2\times 2$ rank-1 PSD  matrix (see Figure \ref{fig:1bit landscape}(b)).
Thus, we factorize $\mX$ into $\mU\mV^\T$ and solve the following optimization problem over the $n\times r$ and $m\times r$ matrices $\mU$ and $\mV$:
\begin{align}
\minimize_{\mU,\mV}\rho_{\Omega,\mY}(\mU,\mV) = f_{\Omega,\mY}(\mU\mV^\T).
\label{eq:1bit fact}\end{align}

To evaluate the performance of this factorization approach on 1-bit matrix completion, we generate $n\times r$ matrices $\mU^\diamond$ and $\mV^\diamond$ with entries drawn i.i.d. from a uniform distribution on $[-\frac{1}{2}, \frac{1}{2}]$ and
construct a random $n\times n$ matrix $\mX^\diamond$ with rank $r$. Similar to the setup in \cite{davenport20141},  the matrix is then scaled so that $\|\mX^\diamond\| = 1$. We obtain 1-bit observations $\{Y_{i,j}\}_{(i,j)\in\Omega}$ by adding Gaussian noise of variance $\sigma^2$ and recording the sign of the resulting value \eqref{eq:1bit observ}, where the subset of indices $\Omega$ is chosen at random with $\E|\Omega| = p$. We compare the performance of the factorization approach and the convex approach \cite{davenport20141} over a range of different values of $n$, $p$, $r$ or $\sigma$. Figures~\ref{fig:1bit vary}(a)-(d) show the normalized squared Frobenius norm of the error $\frac{\|\widehat\mX - \mX^\diamond\|}{\|\mX^\diamond\|_F^2}$ (where $\widehat \mX$ denotes the reconstructed matrix) and average the results over 10 draws of Monte Carlo trials. We observe that matrix factorization approach has {slightly better performance than} the convex approach for 1-bit matrix completion~\cite{davenport20141}. {Note that this phenomenon (the factorization approach having better performance) is also observed in \cite{bhaskar20151}. We repeat these experiments but obtaining 1-bit observations with the logistic regression model where $g(x) = \frac{e^x}{1+e^x}$ for \eqref{eq:1bit observ} and display the results in Figure~\ref{fig:1bit vary logistic}.}

%We begin by experimentally demonstrating the objective function $\rho_{\Omega,\mY}$ has nice optimization property. Figures \ref{fig:1bit diff ini}(a) and (b) show the normalized squared Frobenius norm of the error $\frac{\|\mX_t - \mX^\diamond\|}{\|\mX^\diamond\|_F^2}$ at each iteration $t$ with different initialization for both the case of full measurements (i.e., $\Omega = [n]\times [n]$) and the case with partial measurements.

%\begin{figure}[htb!]
%\begin{minipage}{0.48\linewidth}
%\centerline{
%\includegraphics[width=1.75in]{figure/1bit_Landscape}}
%\centering{(a)}
%\end{minipage}
%\hfill
%\begin{minipage}{0.48\linewidth}
%\centerline{
%\includegraphics[width=1.75in]{figure/1bit_LandscapeWmu}}
%\centering{(b)}
%\end{minipage}
%\caption{\label{fig:1bit landscape} Illustration of the objective function (a) $F_{\Omega,\mY}(\mX)$; (b) $f_{\Omega,\mY}(\mX)$. Here the red {\color{red}$\star$} represents the global minima restricted in the area $|x_1|\leq 1$ and $|x_2|\leq 1$ and the black $\diamond$ represents the possible $a$ and $b$ that generate $\mX^\diamond$.}
%\end{figure}

%\begin{figure}[htb!]
%\begin{minipage}{0.48\linewidth}
%\centerline{
%\includegraphics[width=1.9in]{figure/1bit_Vary_IniFullM}}
%\centering{(a)}
%\end{minipage}
%\hfill
%\begin{minipage}{0.48\linewidth}
%\centerline{
%\includegraphics[width=1.9in]{figure/1bit_Vary_Ini}}
%\centering{(b)}
%\end{minipage}
%\caption{\label{fig:1bit diff ini} (a)$\sigma = 0.3$, $n = 200$, $r = 4$, $p = n^2$ (i.e., full measurements); (b) $\sigma = 0.3$, $n = 200$, $r=4$, $p = 0.25n^2$ }
%\end{figure}

\begin{figure}[htb!]
\begin{minipage}{0.48\linewidth}
\centerline{
\includegraphics[width=1.9in]{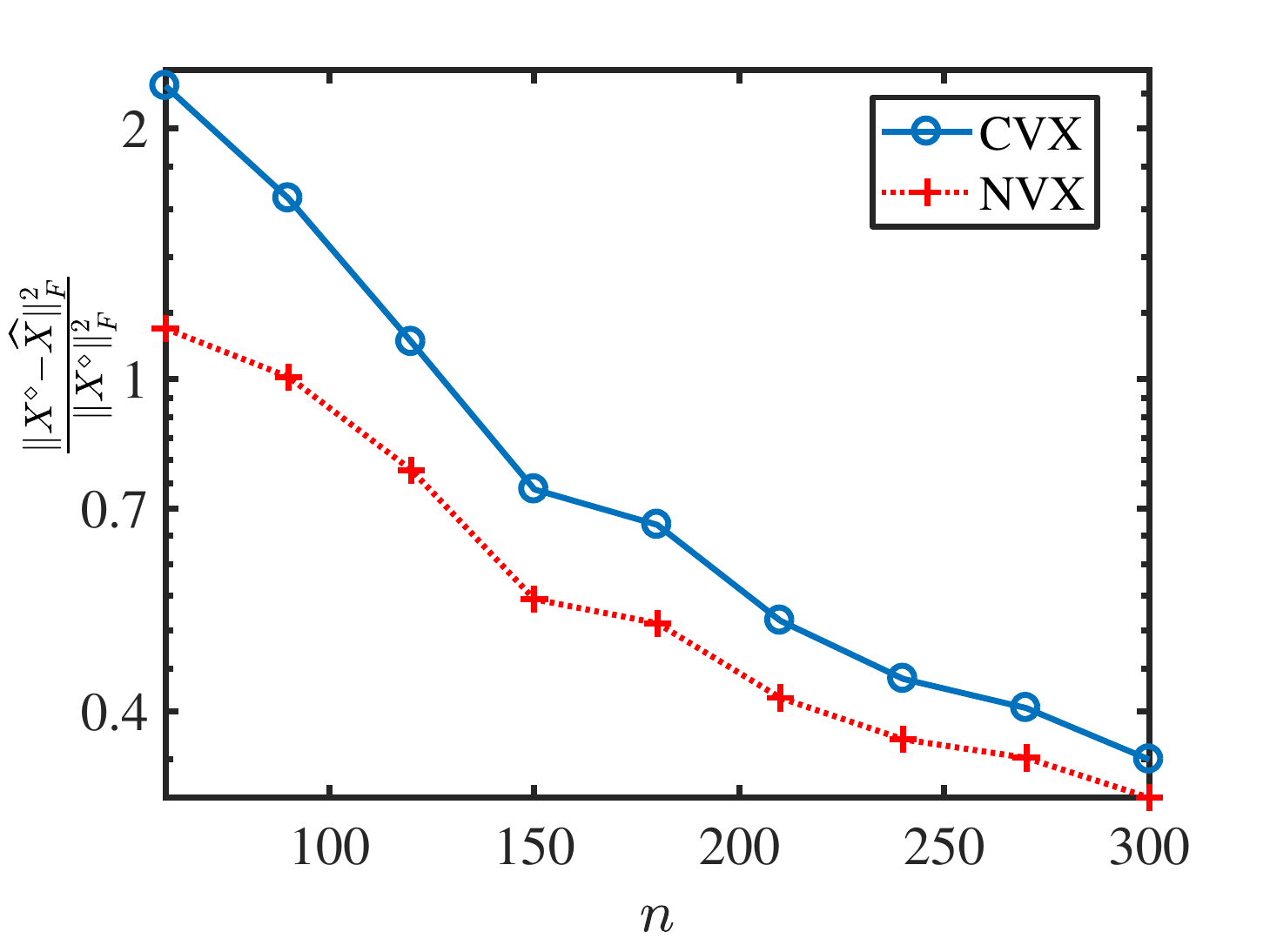}}
\centering{(a)}
\end{minipage}
\hfill
\begin{minipage}{0.48\linewidth}
\centerline{
\includegraphics[width=1.9in]{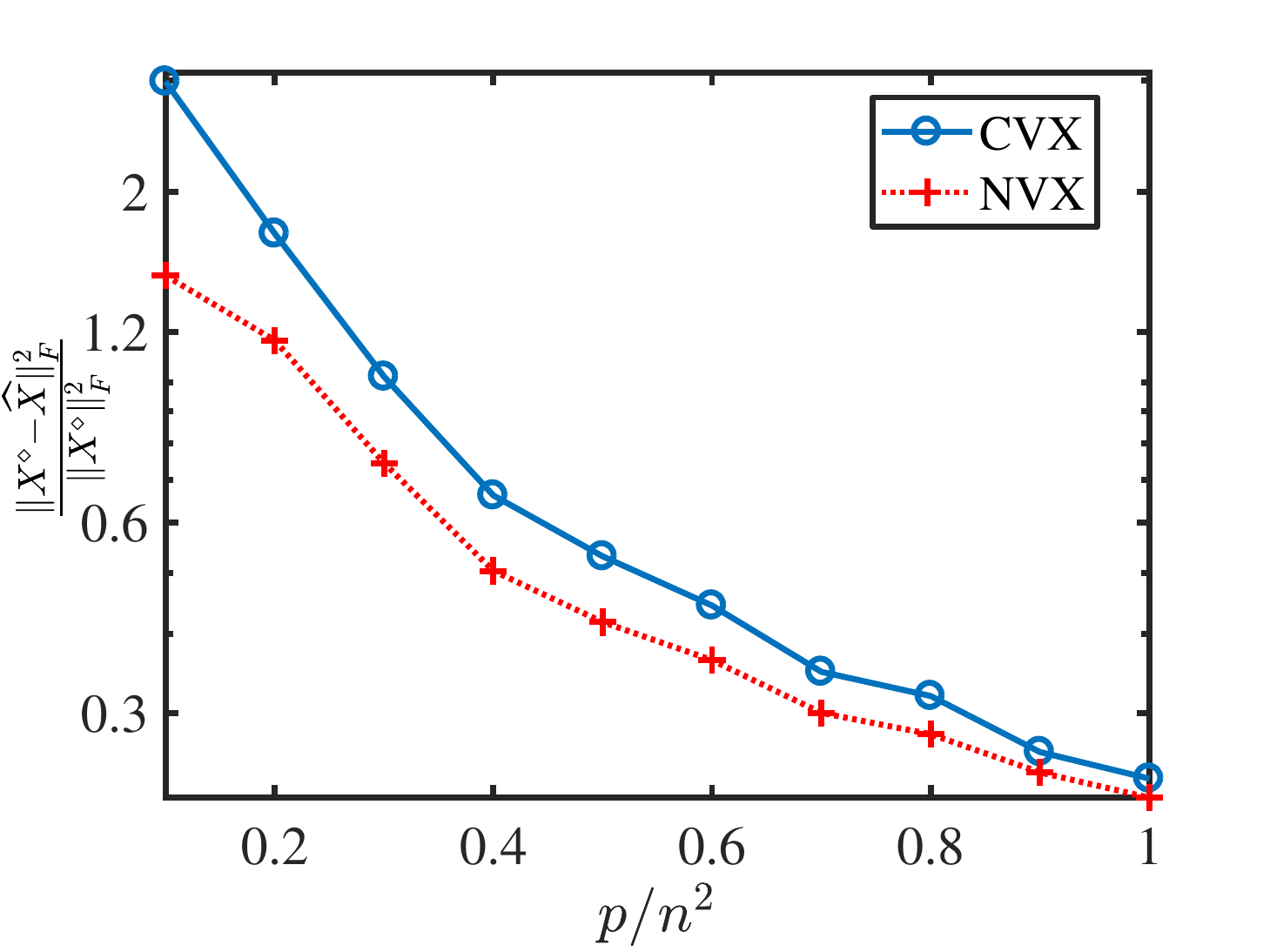}}
\centering{(b)}
\end{minipage}
\vfill
\begin{minipage}{0.48\linewidth}
\centerline{
\includegraphics[width=1.9in]{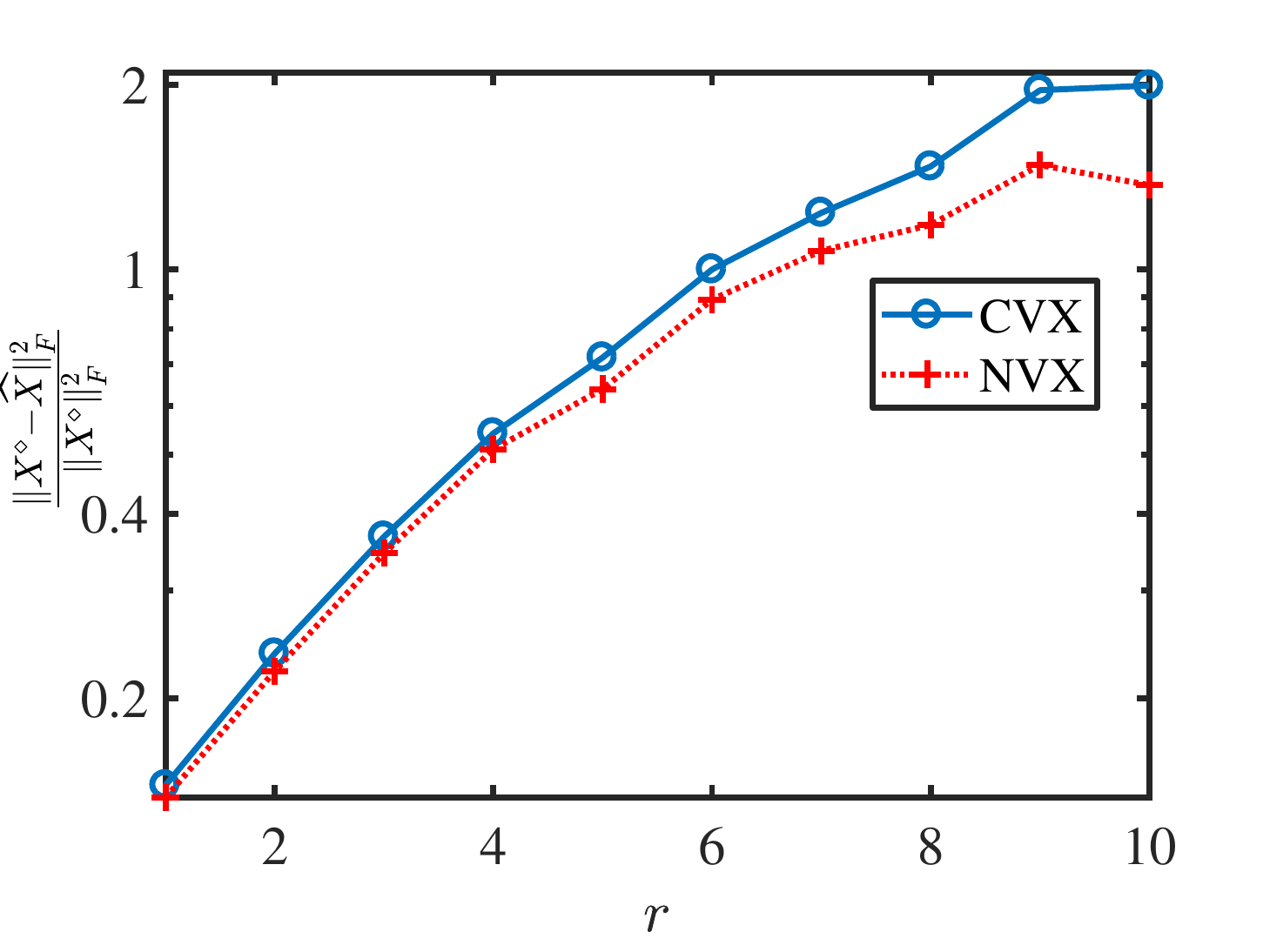}}
\centering{(c)}
\end{minipage}
\hfill
\begin{minipage}{0.48\linewidth}
\centerline{
\includegraphics[width=1.9in]{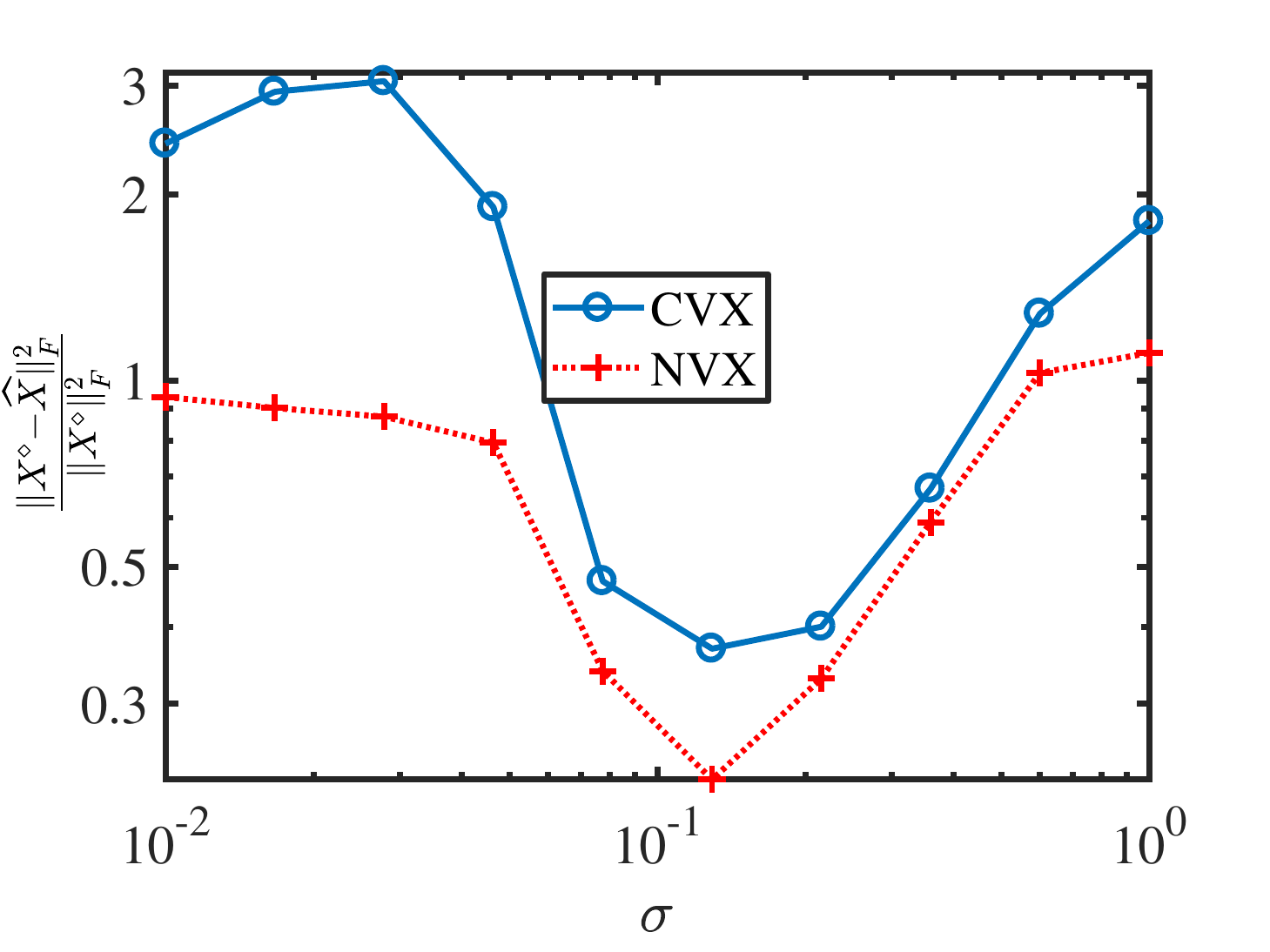}}
\centering{(d)}
\end{minipage}
\caption{\label{fig:1bit vary}The performance in terms of the relative Frobenius norm of the error for the matrix factorization approach (denoted by NVX) and the convex approach in \cite{davenport20141} (denoted by CVX) for solving the 1-bit matrix completion with probit regression model and (a) varying $n$ and $\sigma = 0.3$, $r = 7$, $p = 0.5n^2$; (b) varying $p$ and  $\sigma = 0.3$, $n = 200$, $r = 7$; (c) varying $r$ and $\sigma = 0.3$, $n = 200$, $p = 0.25n^2$; (d) varying $\sigma$ and $n = 200$, $r = 4$, $p = 0.25n^2$. The results are plotted in the log scale.}
\end{figure}

\begin{figure}[htb!]
\begin{minipage}{0.48\linewidth}
\centerline{
\includegraphics[width=1.9in]{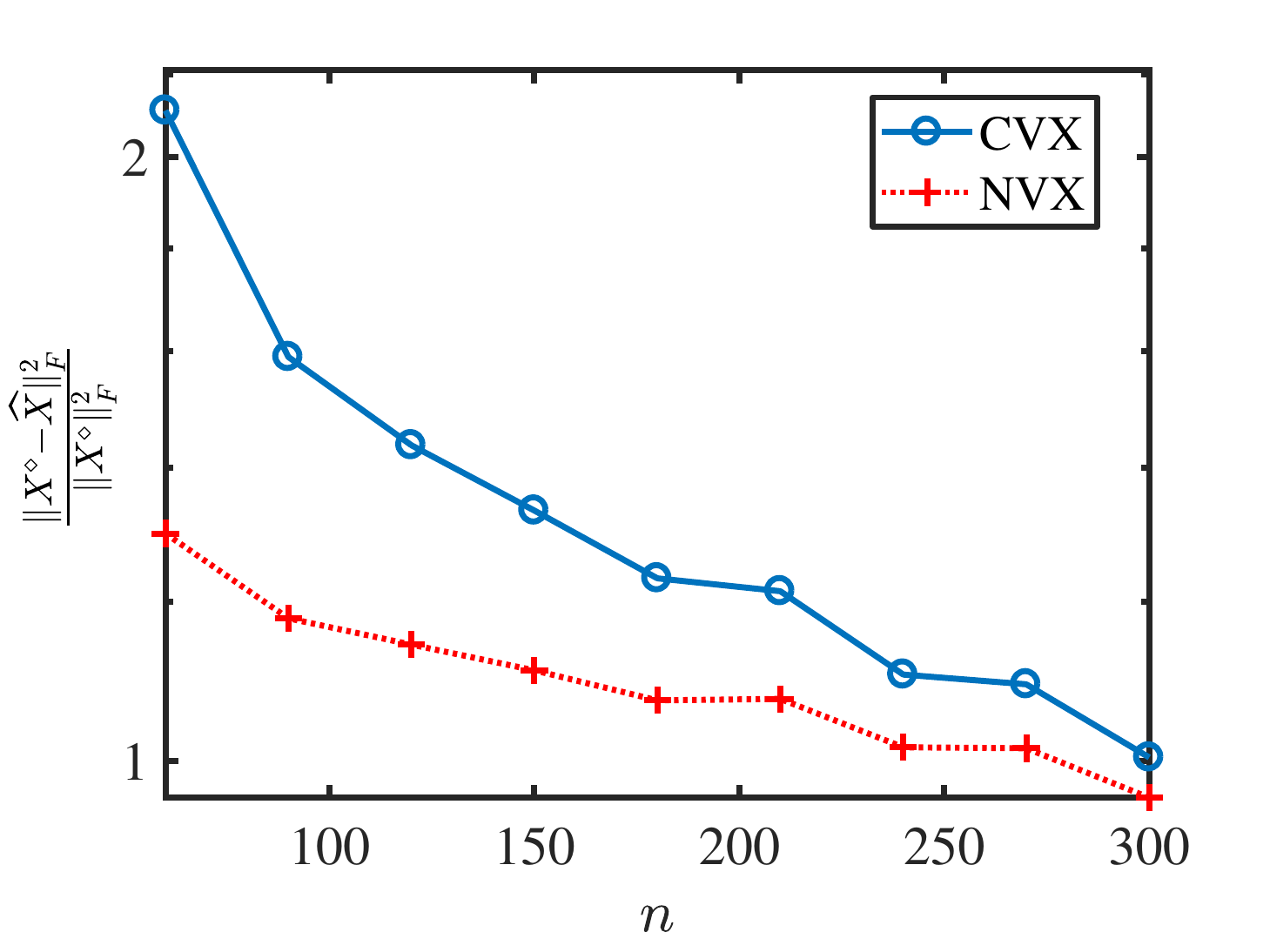}}
\centering{(a)}
\end{minipage}
\hfill
\begin{minipage}{0.48\linewidth}
\centerline{
\includegraphics[width=1.9in]{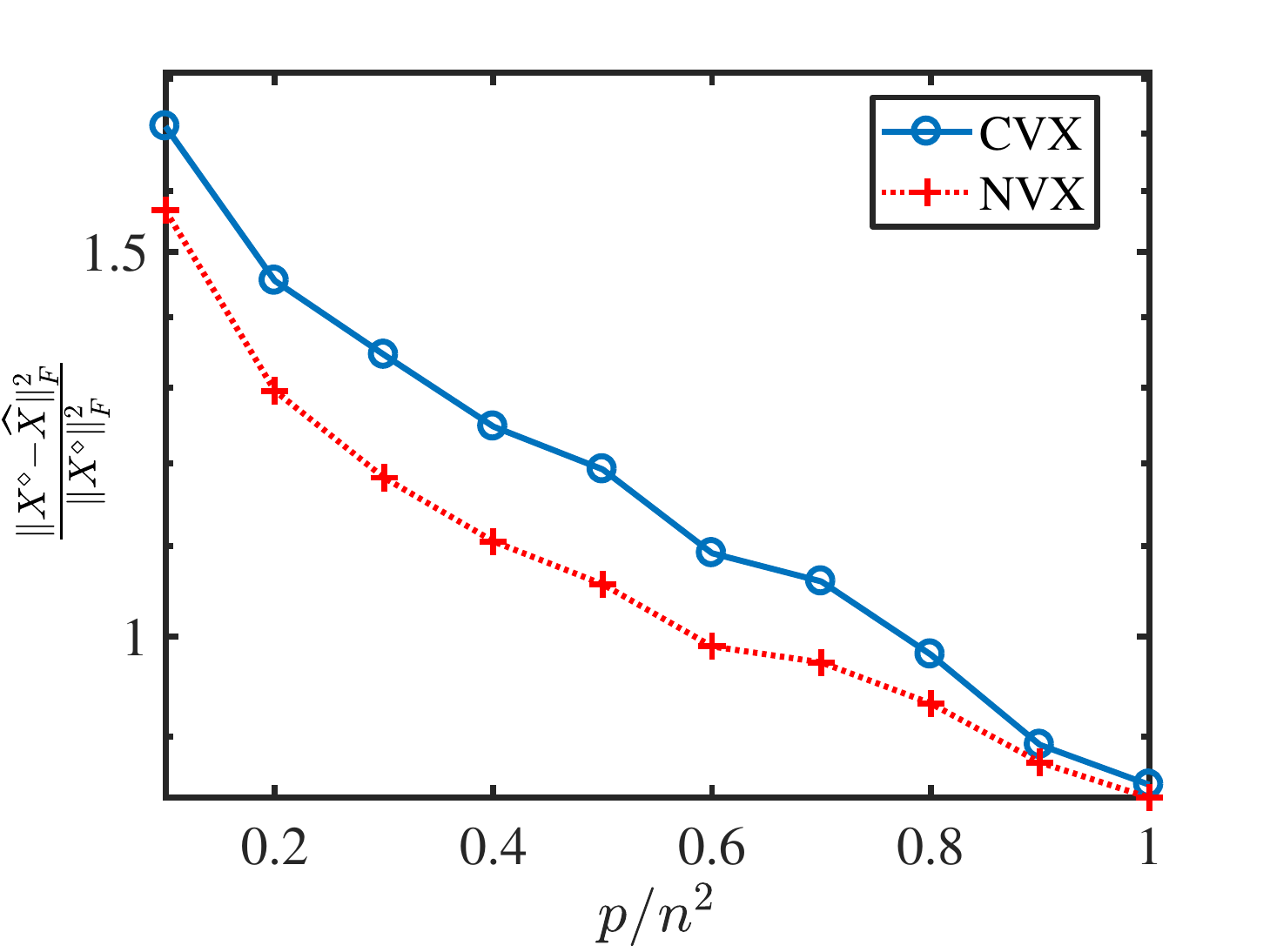}}
\centering{(b)}
\end{minipage}
\caption{\label{fig:1bit vary logistic} The performance in terms of the relative Frobenius norm of the error for the matrix factorization approach (denoted by NVX) and the convex approach in \cite{davenport20141} (denoted by CVX) for solving the 1-bit matrix completion with logistic regression model and (a) varying $n$ and $r = 2$, $p = 0.5n^2$; (b) varying $p$ and  $n = 200$, $r = 2$. The results are plotted in the log scale.}
\end{figure}

\section{Conclusion}\label{sec:conclusion}
This paper considers low-rank matrix optimization {on general (nonsymmetric and rectangular) matrices} with general objective functions. By focusing on general objective functions, we provide a unifying framework for low-rank matrix optimizations with the factorization approach. Although the resulting optimization problem is not convex, we show that the reformulated objection function has a simple landscape: there are no spurious local minima and any critical point not being a local minimum is a strict saddle such that the Hessian evaluated at this point has a strictly negative eigenvalue. These properties guarantee that a number of iterative optimization algorithms (such as gradient descent and the trust region method) will converge to the global optimum from a random initialization.

\appendices

\section{Proof of Lemma~\ref{lem:1bit}}\label{sec:prf 1bit}
\begin{proof}[Proof of Lemma~\ref{lem:1bit}]
We compute the partial derivative of $F_{\Omega,\mY}$ in terms of $X_{i,j}$ as
\begin{align*}
\frac{\partial F_{\Omega,\mY}}{\partial X_{i,j}} &=-\bbmone_{(Y_{i,j}=1)}\frac{q'(X_{i,j})}{q(X_{i,j})} +\bbmone_{(Y_{i,j}=-1)}\frac{q'(X_{i,j})}{1- q(X_{i,j})},
\end{align*}
which implies
\begin{align*}
\frac{\partial^2 F_{\Omega,\mY}}{\partial X_{i,j}\partial X_{i,j}}  =& \bbmone_{(Y_{i,j}=1)}\frac{(q'(X_{i,j}))^2 - q(X_{i,j})q''(X_{i,j})}{q^2(X_{i,j})} \\
+&  \bbmone_{(Y_{i,j}=-1)}\frac{(q'(X_{i,j}))^2 +(1- q(X_{i,j}))q''(X_{i,j})}{(1-q(X_{i,j}))^2}
\end{align*}
and
\begin{align*}
\frac{\partial^2 F_{\Omega,\mY}}{\partial X_{i,j}\partial X_{k,\ell}}  = 0
\end{align*}
for all $(k,\ell)\neq (i,j)$. Thus, the bilinear form for the Hessian of $\nabla^2 F_{\Omega,\mY}(\mX)$ can be computed as
\[
[\nabla^2 F_{\Omega,\mY}(\mX)](\mG,\mG) = \sum_{i}\sum_{j} \frac{\partial^2 F_{\Omega,\mY}}{\partial X_{i,j}\partial X_{i,j}} G^2_{i,j}
\]
for any $\mG\in\R^{n\times m}$. Now since by assumption $\|\mX\|_\infty\leq \gamma$, we have
\[
\alpha_{q,\gamma}\|\mG\|_F^2 \leq [\nabla^2 F_{\Omega,\mY}(\mX)](\mG,\mG) \leq \beta_{q,\gamma} \|\mG\|_F^2.
\]
\end{proof}

\section{Proof of Proposition~\ref{prop:RIP like}}\label{sec:prf RIP like}
\begin{proof}[Proof of Proposition~\ref{prop:RIP like}]
This proof follows similar steps to the proof of~\cite[Lemma 2.1]{candes2008restricted}. First note that  the bilinear form $[\nabla^2f(\mZ)](\mG,\mH) = \sum_{i,j,k,l}\frac{\partial^2 f(\mZ)}{\partial \mZ_{ij}\partial \mZ_{kl}}\mG_{ij}\mH_{kl}$ implies $[\nabla^2f(\mZ)](\mG,\mH)$ is invariant under all scalings for both $\mG$ and $\mH$, i.e.,
\[
[\nabla^2f(\mZ)](a\mG,b\mH) = ab [\nabla^2f(\mZ)](\mG,\mH)
\]
for any $a,b\in\R$. If either $\mG$ or $\mH$ is zero, \eqref{eq:RIP like} holds since both sides are $0$.

Now suppose both $\mG$ or $\mH$ are nonzero. By the scaling invariance property of both sides in \eqref{eq:RIP like}, we assume $\|\mG\|_F = \|\mH\|_F = 1$ without loss of generality. Note that the $(2r,4r)$-restricted  strong convexity and smoothness condition~\eqref{eq:RIP like} implies
\begin{align*}
\alpha \left\|\mG \pm \mH\right\|_F^2 &\leq [\nabla^2 f(\mX)](\mG\pm\mH,\mG\pm \mH)\\ &\leq \beta \left\|\mG \pm \mH\right\|_F^2.
\end{align*}
Thus we have
\begin{align*}
&-\frac{\beta-\alpha}{2} \left(\left\|\mG\right\|_F^2 +\left\|\mH\right\|_F^2\right)\\
&\leq 2\left[\nabla^2f(\mZ)\right](\mG,\mH) - (\alpha+\beta)\left\langle \mG,\mH \right\rangle\\
&\leq \frac{\beta-\alpha}{2} \left(\left\|\mG\right\|_F^2 +\left\|\mH\right\|_F^2\right),
\end{align*}
which further implies
\begin{align*}
&\left| 2\left[\nabla^2f(\mZ)\right](\mG,\mH) - (\alpha+\beta)\left\langle \mG,\mH \right\rangle\right| \\
&\leq \beta-\alpha = (\beta-\alpha)\left\|\mG\right\|_F \left\|\mH\right\|_F.
\end{align*}
\end{proof}

\section{Proof of Lemma \ref{lem:bound:WW}}\label{sec:proof bound WW - W*W*QQ}
\begin{proof}[Proof of Lemma \ref{lem:bound:WW}]
First recall the notation $\mX = \mU\mV^\T$, $\mX^\star = \mU^\star\mV^\star $, and
\[
\mW = \begin{bmatrix} \mU \\ \mV \end{bmatrix}, \ \widehat\mW = \begin{bmatrix} \mU \\ -\mV \end{bmatrix}, \mW^\star = \begin{bmatrix} \mU^\star \\ \mV^\star \end{bmatrix}, \ \widehat\mW^\star = \begin{bmatrix} \mU^\star \\ -\mV^\star \end{bmatrix}.
\]
It follows from \eqref{eq:proof thm cirtical 3} and \eqref{eq:proof thm cirtical 4} that any critical point $\mW$ satisfies
\[
\begin{bmatrix}\mzero & \nabla f(\mX) \\ \nabla f(\mX)^\T & \mzero\end{bmatrix}\mW =\mzero ,
\]
which gives
\begin{align}\label{eq:proof bound WW - W*W*QQ 0}\begin{split}
&0 = \langle \begin{bmatrix}\mzero & \nabla f(\mX) \\ \nabla f(\mX)^\T & \mzero\end{bmatrix}, \mZ \mW^\T \rangle\\
&=\langle \begin{bmatrix}\mzero & \nabla f(\mX) - \nabla f(\mX^\star) \\ \nabla f(\mX)^\T - \nabla f(\mX^\star)^\T & \mzero\end{bmatrix}, \mZ \mW^\T \rangle\\
& = \underbrace{\langle  \nabla f(\mX) - \nabla f(\mX^\star) - \frac{\alpha + \beta}{2}(\mX -\mX^\star ) , \mZ_{\mU} \mV^\T + \mU\mZ_{\mV}^\T \rangle}_{\daleth_1}\\
& \quad + \frac{\alpha + \beta}{2}\underbrace{\left\langle \mX -\mX^\star, \mZ_{\mU} \mV^\T + \mU\mZ_{\mV}^\T \right\rangle}_{\daleth_2}
\end{split}\end{align}
for any $\mZ = \begin{bmatrix}\mZ_{\mU}\\\mZ_{\mV}  \end{bmatrix}\in\R^{(n+m)\times r}$. Here the second line utilizes the fact $\nabla f(\mX^\star)=\mzero$. We bound $\daleth_1$ by first using integral form of the mean value theorem for $\nabla f(\mX)$:
\begin{align*}
&\daleth_1 =\\&  \int_0^1\left[ \nabla^2 f(t\mX + (1-t)\mX^\star)\right](\mX -\mX^\star, \mZ_{\mU} \mV^\T + \mU\mZ_{\mV}^\T) d  t \\ & \quad- \frac{\alpha + \beta}{2}\left\langle \mX -\mX^\star , \mZ_{\mU} \mV^\T + \mU\mZ_{\mV}^\T \right\rangle.
\end{align*}
Noting that all the three matrices $t\mX + (1-t)\mX^\star$, $\mX -\mX^\star$ and $\mZ_{\mU} \mV^\T + \mU\mZ_{\mV}^\T$ have rank at most $2r$, it follows from Proposition~\ref{prop:RIP like} that
\begin{align*}
\left|\daleth_1 \right|
&\leq \frac{\beta - \alpha}{2}\left\|\mX -\mX^\star\right\|_F\left\|\mZ_{\mU} \mV^\T + \mU\mZ_{\mV}^\T\right\|_F,
\end{align*}
which when plugged into \eqref{eq:proof bound WW - W*W*QQ 0} gives
\begin{equation}\begin{split}
&\frac{\alpha + \beta}{2}\daleth_2 = -\daleth_1 \\
&\leq \frac{\beta - \alpha}{2}\left\|\mX -\mX^\star\right\|_F\left\|\mZ_{\mU} \mV^\T + \mU\mZ_{\mV}^\T\right\|_F.
\end{split}\label{eq:proof bound WW - W*W*QQ 1}
\end{equation}

Now let $\mZ = (\mW\mW^\T - \mW^\star\mW^{\star\T}){\mW^T}^{\dagger}$, which gives $\mZ\mW^\T = (\mW\mW^\T - \mW^\star\mW^{\star\T})\mP_{\mW}$. Here $\dagger$ denotes the pseudoinverse of a matrix and $\mP_{\mW}$ is the orthogonal projector onto the range of $\mW$. Utilizing the fact $\widehat\mW^\T\mW = \mzero$ from~\eqref{eq:thm eq 1}, we further connect the left hand side of \eqref{eq:proof bound WW - W*W*QQ 1} with $\left\| \left(\mW\mW^\T - \mW^\star\mW^{\star\T}\right)\mP_{\mW}
\right\|_F^2$ by
\begin{equation}\label{eq:proof bound WW - W*W*QQ 2}
\begin{split}
&\frac{\alpha + \beta}{2}\daleth_2 = \frac{\alpha + \beta}{2}\daleth_2 + \frac{\alpha + \beta}{4}\langle\widehat\mW\widehat\mW^\T,\mZ\mW^\T  \rangle\\
& = \frac{\alpha + \beta}{4} \left\langle \mW\mW^\T -  \mW^\star\mW^{\star\T}, \left(\mW\mW^\T - \mW^\star\mW^{\star\T}\right)\mP_{\mW}
\right\rangle \\
& \quad + \frac{\alpha + \beta}{4} \left\langle   \widehat\mW^\star\widehat\mW^{\star\T}, \left(\mW\mW^\T - \mW^\star\mW^{\star\T}\right)\mP_{\mW}
\right\rangle\\
& \geq \frac{\alpha + \beta}{4} \left\langle \mW\mW^\T -  \mW^\star\mW^{\star\T}, \left(\mW\mW^\T - \mW^\star\mW^{\star\T}\right)\mP_{\mW}
\right\rangle \\
&=
 \frac{\alpha + \beta}{4} \left\| \left(\mW\mW^\T - \mW^\star\mW^{\star\T}\right)\mP_{\mW}
\right\|_F^2,
\end{split}\end{equation}
where the inequality follows because $\left\langle   \widehat\mW^\star\widehat\mW^{\star\T}, \mW^\star\mW^{\star\T}\mP_{\mW}
\right\rangle = 0$ (noting that $\widehat\mW^{\star\T}\widehat\mW^{\star}= \mzero$) and $\left\langle   \widehat\mW^\star\widehat\mW^{\star\T}, \mW\mW^\T\mP_{\mW}
\right\rangle = \left\langle   \widehat\mW^\star\widehat\mW^{\star\T}, \mW\mW^\T
\right\rangle \geq  0$ since it is the inner product between two PSD matrices.

%\begin{align*}
%2\daleth_2 = \left\langle \mW\mW^\T -  \mW^\star\mW^{\star\T} - ( \widetilde \mW \widetilde \mW^\T - \widetilde\mW^\star \widetilde \mW^{\star\T})    ,\left(\mW\mW^\T - \mW^\star\mW^{\star\T}\right)\mQ_{\mW}\mQ_{\mW}^\T
%\right\rangle,
%\end{align*}
%and
%\begin{align*}
%\left\langle \widehat\mW\widehat\mW^\T,\mZ\mW^\T  \right\rangle \geq  \left\langle  \widetilde \mW \widetilde \mW^\T - \widetilde\mW^\star \widetilde \mW^{\star\T}, \left(\mW\mW^\T - \mW^\star\mW^{\star\T}\right)\mQ_{\mW}\mQ_{\mW}^\T
%\right\rangle.
%\end{align*}
On the other hand, we give an upper bound on the right hand side of \eqref{eq:proof bound WW - W*W*QQ 1}:
\begin{align*}
&\left\|\mX -\mX^\star\right\|_F  \left\|\mZ_{\mU} \mV^\T + \mU\mZ_{\mV}^\T\right\|_F \\&\leq \left\|\mX -\mX^\star\right\|_F\sqrt{2\left\|\mZ_{\mU} \mV^\T\right\|_F^2 +2 \left\|\mU\mZ_{\mV}^\T\right\|_F^2}\\
&\leq \left\|\mX -\mX^\star\right\|_F  \left\| \left(\mW\mW^\T - \mW^\star\mW^{\star\T}\right)\mP_{\mW}
\right\|_F,
\end{align*}
where the last line follows because $\left\|\mZ_{\mU} \mV^\T\right\|_F^2 + \left\|\mZ_{\mV}\mU^\T\right\|_F^2 = \left\|\mZ_{\mU} \mU^\T\right\|_F^2 + \left\|\mZ_{\mV}\mV^\T\right\|_F^2$ (since $\mU^\T\mU = \mV^\T\mV$), implying $2\left\|\mZ_{\mU} \mV^\T\right\|_F^2 +2 \left\|\mU\mZ_{\mV}^\T\right\|_F^2 = \left\|\mZ\mW^\T\right\|_F^2$. This together with \eqref{eq:proof bound WW - W*W*QQ 1} and \eqref{eq:proof bound WW - W*W*QQ 2} completes the proof.
\end{proof}

\section{Proof of Lemma \ref{lem:CC - DD to WDelta}}
\label{sec:proof CC - DD to WDelta}
\begin{proof}[Proof of Lemma \ref{lem:CC - DD to WDelta}]
When $\mC\neq \mzero$, the proof follows directly from the following results.
\begin{lem}\label{lem:CC - DD to WDelta 1}\cite[Lemma 2]{li2016} For any matrices $\mC,\mD \in\R^{n\times r}$ with rank $r_1$ and $r_2$, respectively,  let $\mR = \argmin_{\widetilde \mR\in\calO_r}\|\mC - \mD\mR\|_F$. Then
 \begin{align*}
 \left\|\mC\mC^\T - \mD\mD^\T \right\|_F \geq \min\left\{\sigma_{r_1}(\mC), \sigma_{r_2}(\mD)\right\}\cdot \left\|\mC - \mD\mR\right\|_F.
\end{align*}
 \end{lem}

\begin{lem}\label{lem:CC - DD to WDelta 2}\cite[Lemma 5.4]{tu2015low} For any matrices $\mC,\mD \in\R^{n\times r}$ with $\rank(\mD) =r$,  let $\mR = \argmin_{\widetilde \mR\in\calO_r}\|\mC - \mD\mR\|_F$. Then
 \begin{align*}
 \left\|\mC\mC^\T - \mD\mD^\T \right\|_F^2 \geq 2(\sqrt{2}-1)\sigma_r^2(\mD) \left\|\mC - \mD\mR\right\|_F^2.
\end{align*}
 \end{lem}

If $\mC = \mzero$, then we have
\begin{align*}
& \left\|\mC\mC^\T - \mD\mD^\T \right\|_F^2 =  \left\|\mD\mD^\T \right\|_F^2 = \sum_{i=1}^{r_2}\sigma_i^4(\mD)\\& \geq \sigma_{r_2}^2(\mD)\sum_{i=1}^{r_2}\sigma_i^2(\mD)  = \sigma_{r_2}^2(\mD) \left\|\mC - \mD\mR\right\|_F^2.
\end{align*}
\end{proof}

\section{Proof of~\eqref{eq:bound Pi}}\label{sec:proof eq bound Pi}
\begin{proof}[Proof of~\eqref{eq:bound Pi}] We prove the upper bounds for the four terms as follows.

{\noindent\bf Bounding term $\Pi_1$:}
Utilizing the fact that $\mDelta_{\mU} = \mU- \mU^\star\mR$ and $\mDelta_{\mV} = \mV- \mV^\star\mR$, we have
\begin{align*}
\Pi_1& =\left\langle \nabla f(\mX),\mDelta_{\mU}\mDelta_{\mV}^\T \right\rangle \\
&= \left\langle \nabla f(\mX), (\mU- \mU^\star\mR)(\mV- \mV^\star\mR)^\T\right\rangle\\
& = \left\langle \nabla f(\mX), \mX + \mX^\star -\mU^\star\mR^\T\mV^T - \mU\mR^\T{\mV^\star}^\T\right\rangle \\
&\stackrel{(i)}{=} -\left\langle \nabla f(\mX), \mX -\mX^\star\right \rangle \\
&\stackrel{(ii)}{=} -\left\langle \nabla f(\mX) - \nabla f(\mX^\star), \mX -\mX^\star \right\rangle \\
&\stackrel{(iii)}{\leq} - \alpha \left\| \mX -\mX^\star\right\|_F^2,
\end{align*}
where $(i)$ follows from \eqref{eq:proof thm cirtical 3} and \eqref{eq:proof thm cirtical 4},
$(ii)$ utilizes $\nabla f(\mX^\star)=\mzero$, and $(iii)$ follows by using the $(2r,4r)$-restricted strict convexity property \eqref{eq:RIP like}:
\begin{align*}
&\left\langle \nabla f(\mX) - \nabla f(\mX^\star), \mX -\mX^\star \right\rangle\\
&= \int_0^1\left[ \nabla^2 f(t\mX + (1-t)\mX^\star)\right]\left(\mX -\mX^\star, \mX -\mX^\star\right) d  t\\
&\geq  \int_0^1\alpha \left\langle \mX -\mX^\star, \mX -\mX^\star\right\rangle d  t\\
&= \alpha \left\|\mX - \mX^\star\right\|_F^2,
\end{align*}
where the first line follows from the integral form of the mean value theorem for vector-valued functions, and the second line uses the fact that both $t\mX + (1-t)\mX^\star$ and $\mX -\mX^\star$ have rank at most $2r$,  and the $(2r,4r)$-restricted strong convexity of the Hessian $\nabla^2 f(\cdot)$.

\vspace{.1in}
{\noindent\bf Bounding term $\Pi_2$:}
By the smoothness condition \eqref{eq:RIP like}, we have
\begin{align*}
\Pi_2& = \left[\nabla^2f(\mX)\right]\left(\mDelta_{\mU}\mV^\T+\mU\mDelta_{\mV}^\T,\mDelta_{\mU}\mV^\T+\mU\mDelta_{\mV}^\T\right)\\
&\leq \beta\left\|\mDelta_{\mU}\mV^\T+\mU\mDelta_{\mV}^\T\right\|_F^2\\
&\leq  2\beta\left(\left\|\mDelta_{\mU}\mV^\T\right\|_F^2+\left\|\mU\mDelta_{\mV}^\T\right\|_F^2\right)\\
& = \beta \left\|\mW\mDelta^\T\right\|_F^2,
\end{align*}
where the last line holds because $ \left\|\mD\mU^\T\right\|_F^2  = \left\|\mD\mV^\T\right\|_F^2$
for any $\mD\in\R^{p\times r}$ with arbitrary $p\geq 1$ since any critical point $\mW$ satisfies $\mU^\T\mU = \mV^\T\mV$.

\vspace{.1in}
{\noindent \bf Bounding term $\Pi_3$:}
\begin{align*}
&\Pi_3\\ &= \langle \mU\mDelta_{\mU}^\T,\mDelta_{\mU}\mU^\T\rangle + \langle\mV\mDelta_{\mV}^\T,\mDelta_{\mV}\mV^\T \rangle- 2 \langle\mU\mDelta_{\mV}^\T,\mDelta_{\mU}\mV^\T \rangle\\
&\leq\left \|\mU\mDelta_{\mU}^\T\right\|_F^2 + \left\|\mV\mDelta_{\mV}^\T\right\|_F^2 + \left\|\mU\mDelta_{\mV}^\T\right\|_F^2 +\left\|\mV\mDelta_{\mU}^\T\right\|_F^2 \\
& = \left\|\mW\mDelta^\T\right\|_F^2.
\end{align*}

\vspace{.1in}
{\noindent \bf Bounding term $\Pi_4$:}
\begin{align*}
\Pi_4 &= \left\langle \widehat \mW\widehat\mW^\T, \left(\mW - \mW^\star\mR\right)\left(\mW - \mW^\star\mR\right)^\T \right\rangle \\
&\stackrel{(i)}{=}-\left\langle \widehat \mW\widehat\mW^\T, \mW\mW^\T -\mW^\star{\mW^\star}^\T \right\rangle\\
&\stackrel{(ii)}{\leq} -\left\langle \widehat \mW\widehat\mW^\T, \mW\mW^\T -\mW^\star{\mW^\star}^\T \right\rangle \\ &\quad +\left\langle \widehat \mW^\star\widehat\mW^{\star\T}, \mW\mW^\T -\mW^\star{\mW^\star}^\T \right\rangle\\
& = -\left\langle \widehat \mW\widehat\mW^\T - \widehat \mW^\star\widehat\mW^{\star\T}, \mW\mW^\T -\mW^\star{\mW^\star}^\T \right\rangle\\
%&= 2\left\|\mX - \mX^\star\right\|_F^2 - \left\|\mU\mU^\T - \mU^\star\mU^{\star\T}\right\|_F^2 - \left\|\mV\mV^\T - \mV^\star\mV^{\star\T}\right\|_F^2\\
& \leq 2\left\|\mX - \mX^\star\right\|_F^2,
\end{align*}
where $(i)$ holds because $\widehat\mW^\T\mW = \mzero$, and $(ii)$ follows because  $\widehat\mW^{\star\T} \mW^\star = \mzero$ and $\langle \widehat \mW^\star\widehat\mW^{\star\T}, \mW\mW^\T \rangle\geq 0$ since it is the inner product between two PSD matrices.
\end{proof}

\section{Proof of~\eqref{eq:bound 5}}\label{sec:prf eq bound 5}
\begin{proof}[Proof of~\eqref{eq:bound 5}]
To show~\eqref{eq:bound 5}, expanding the left hand side of~\eqref{eq:bound 5}, it is equivalent to show
\[
\left\|\mU\mU^\T - \mU^{\star}\mU^{\star\T} \right\|_F^2 + \left\|\mV\mV^\T - \mV^{\star}\mV^{\star\T} \right\|_F^2 \leq 2\left\|\mX - \mX^\star\right\|_F^2.
\]
Expanding both sides of the above equation and utilizing the fact $\mU^\T\mU = \mV^\T\mV$ and $\mU^{\star\T}\mU^\star = \mV^{\star\T}\mV^\star$, the remaining step is to show
\begin{align*}
&\trace\left(\mU\mU^\T\mU^{\star}\mU^{\star\T}\right) + \left( \mV\mV^\T\mV^{\star}\mV^{\star\T}\right) \\ & \geq 2 \trace\left(\mU\mV^\T\mV^{\star}\mU^{\star\T}\right).
\end{align*}
Thus, we obtain~\eqref{eq:bound 5} by noting that the above equation is equivalent to
\[
\trace\left(\left(\mU^{\star\T}\mU -\mV^{\star\T}\mV\right)^2\right)\geq 0.
\]
\end{proof}

\bibliographystyle{ieeetr}
\bibliography{nonconvex}
\end{document}